\documentclass[onecolumn,amsmath,amssymb,10pt,aps]{revtex4}
  
\usepackage[utf8]{inputenc}
\usepackage[T1]{fontenc}

\usepackage{amsmath}
\usepackage{amsfonts}
\usepackage{graphicx}
\usepackage{yfonts}
\usepackage{color}
\usepackage[normalem]{ulem}
\usepackage{amsthm}
\usepackage{bm}
\usepackage{bbm}
\usepackage{mathtools}
\usepackage{array}
\usepackage{placeins}
\usepackage{enumitem}
\usepackage{mdframed}

\newcommand\bigforall{\mbox{\Large $\mathsurround=1pt\forall$}}

\def\<{\langle}
\def\>{\rangle}

\newcommand{\Tr}{\mathrm{Tr}}
\def\oper{{\mathchoice{\rm 1\mskip-4mu l}{\rm 1\mskip-4mu l}
{\rm 1\mskip-4.5mu l}{\rm 1\mskip-5mu l}}}
\DeclareMathAlphabet\mathbfcal{OMS}{cmsy}{b}{n}

\mathchardef\mhyphen="2D 

\newtheorem{Definition}{Definition}

\newtheorem{Proposition}{Proposition}
\newtheorem{Example}{Example}

\usepackage{thmtools}

\definecolor{aliceblue}{rgb}{0.94, 0.97, 1.0}

\declaretheorem[
  shaded={rulecolor=black, rulewidth=1pt, bgcolor=aliceblue},
  name=Theorem,
]{Theorem}

\begin{document}

\title{Equivalence relations between conical 2-designs and mutually unbiased generalized equiangular tight frames}

\author{Katarzyna Siudzi\'{n}ska\footnote{e-mail: kasias@umk.pl}}
\affiliation{Institute of Physics, Faculty of Physics, Astronomy and Informatics, Nicolaus Copernicus University in Toru\'{n}, ul.~Grudzi\k{a}dzka 5, 87--100 Toru\'{n}, Poland}

\begin{abstract}
Quantum measurements play a fundamental role in quantum information. Therefore, increasing efforts are being made to construct symmetric measurement operators for qudit systems. A wide class of projective measurements corresponds to complex projective 2-designs, which include symmetric, informationally complete (SIC) POVMs and mutually unbiased bases (MUBs). In this paper, we establish a one-to-one correspondence between conical 2-designs and mutually unbiased generalized equiangular tight frames, both of which are common generalizations of SIC POVMs and MUBs to operators of arbitrary rank. It turns out that there exist rich families of operators that belong to only one of those two classes. This raises important questions about which symmetries have to be preserved for applicational prominence.
\end{abstract}

\flushbottom

\maketitle

\thispagestyle{empty}

\section{Introduction}

Quantum measurements are important tools in quantum information theory for the purpose of extracting information from unknown quantum states. Their various applications include e.g. quantum communication protocols \cite{Zhou,Song,Bouchard}, quantum state estimation \cite{IOC1,Adamson,Zhu_QSE}, entanglement detection \cite{ESIC,EW-SIC,KalevBae,Blume,SM_Pmaps}, steering \cite{Lai2,OperationalSICs,Bene}, and quantum tomography \cite{Prugovecki,PetzRuppert,Pimenta,Bent}. The most well known measurements are symmetric, informationally complete (SIC) POVMs \cite{Renes} and mutually unbiased bases (MUBs) \cite{Schwinger,Szarek}, which are examples of complex projective 2-designs \cite{Neumaier,Hoggar,Zauner,Scott}.

In recent years, much attention has been paid to formulate generalized quantum measurements that consist SIC POVMs or MUBs as special cases. One approach introduces semi-SIC POVMs \cite{semi-SIC} and equioverlapping measurements \cite{EOM22,EOM24,EOMq3,EOM25} -- projective measurements that go beyond complex projective 2-designs. Alternatively, one defines families of mutually unbiased symmetric measurements with arbitrary rank. $(N,M)$-POVMs describe collections of $N$ mutually unbiased POVMs, where each POVM consists of $M$ operators \cite{SIC-MUB}. Generalized symmetric measurements allow for these POVMs to have different numbers of elements \cite{SIC-MUB_general}, whereas generalized equiangular measurements (GEAMs) assume that each POVM can be taken with a different weight \cite{GEAM}. It turns out that every $(N,M)$-POVM and generalized symmetric measurement forms a conical 2-design \cite{SICMUB_design,SIC-MUB_general}, which is a generalization of complex projective 2-designs to non-projective operators \cite{Graydon,Graydon2,GraydonPhD}. This has not been observed for generalized equiangular measurements, where large families exist that go beyond conical 2-designs. However, many important discoveries could be made only for GEAMs that are conical 2-designs, like the relation between the index of coincidence and state purity \cite{GEAM} that is necessary for the derivation of entropic uncertainty relations, quantum coherence, the Brukner-Zeilinger invariants, and entanglement detection \cite{GEAM_Pmaps,GEAM_coherence}. A natural question arises: are there POVMs that form conical 2-designs but are not GEAMs? If so, can they be easily characterized?

In this paper, we analyze conical 2-designs with distinct, non-zero elements, as well as their relations to mutually unbiased generalized equiangular tight frames (MU GETFs), which are simply GEAMs rescaled by a common positive factor. Our main results are establishing the necessary and sufficient conditions for conical 2-designs to be MU GETFs. Analogical equivalence is then derived under the assumption that conical 2-designs are POVMs. Our main results are collected into Theorems 1-4. During our calculations, we find a large family of conical 2-designs that are not MU GETFs. This opens a new direction for further research.

Our paper is organized as follows. In Section 2, we recall the definition and basic properties of conical 2-designs. Section 3 introduces generalized equiangular tight frames (GETFs) as a generalization of equiangular tight frames to operators of arbitrary rank. Section 4 establishes implication relations between GETFs and conical 2-designs constructed from linearly independent operators. It also gives the most general trace conditions and a one-to-one correspondence between a class of conical 2-designs and GETFs. Next, we consider collections of linearly dependent operators that span the operator space. In Section 5, we define mutually unbiased (MU) GETFs as a straightforward generalization of generalized equiangular measurements. Then, in an analogy to Section 4, Section 6 compares a class of informationally overcomplete conical 2-designs with MU GETFs. We provide a set of conditions under which the sets of GEAMs and conical 2-designs are equivalent. Our results are summarized in Conclusions, where we also present a list of open questions for further considerations.

\section{Conical 2-designs}

Conical designs have been introduced as a generalization of complex projective designs from rank-1 projectors to semi-positive operators of arbitrary rank. Throughout this paper, we assume that the operators act on the Hilbert space $\mathcal{H}\simeq\mathbb{C}^d$.

\begin{Definition}(\cite{Graydon})
A conical 2-design $\mathcal{R}=\{R_k;\,k=1,\ldots,M\}$ is a collection of $M\geq d^2$ semi-positive operators $R_k$ such that
\begin{equation}\label{conic}
\sum_{k=1}^MR_k\otimes R_k=\kappa_+\mathbb{I}_d\otimes\mathbb{I}_d+\kappa_-\mathbb{F}_d,
\end{equation}
where $\kappa_+\geq\kappa_->0$, and  $\mathbb{F}_d=\sum_{m,n=1}^d|m\>\<n|\otimes|n\>\<m|$ is the flip operator.
\end{Definition}

From definition, $\sum_{k=1}^MR_k\otimes R_k$ commutes with $U\otimes U$ for any unitary $U$, and $R_k$ span the operator space $\mathcal{B}(\mathcal{H})$ on $\mathcal{H}$. Moreover, in ref. \cite{SIC-MUB_general}, it has been shown that, through the Choi-Jamio{\l}kowski isomorphism \cite{Choi,Jamiolkowski}, eq. (\ref{conic}) is equivalent to
\begin{equation}\label{kap}
\Phi=\kappa_+d\Phi_0+\kappa_-\oper.
\end{equation}
In the above formula, $\Phi_0[X]=\mathbb{I}_d\Tr(X)/d$ is the maximally depolarizing channel, $\oper$ denotes the identity map, and $\Phi[X]=\sum_{k=1}^MR_k\Tr(R_kX)$ is a linear map constructed from $R_k$. Also, note that eq. (\ref{conic}) implies
\begin{equation}\label{conicsum}
\Tr_2\left(\sum_{k=1}^MR_k\otimes R_k\right)=\sum_{k=1}^MR_k\Tr(R_k)=(d\kappa_++\kappa_-)\mathbb{I}_d,
\end{equation}
where $\Tr_2(X)$ is the trace of $X$ over the second subsystem. Moreover, there are special classes of conical 2-designs:
\begin{enumerate}[label=(\it\roman*)]
\item {\it homogeneous conical 2-designs}, for which $\Tr(R_k)$ and $\Tr(R_k^2)$ are constant \cite{Graydon};
\item {\it complex projective 2-designs}, whose elements are rank-1 projectors \cite{Neumaier,Hoggar,Scott}.
\end{enumerate}

\begin{Example}
There exist linear operators that satisfy eq. (\ref{conic}) with $\kappa_+\geq\kappa_->0$ but are not semi-positive. Consider the following qubit operators,
\begin{equation}
\begin{split}
R_1&=\frac{1}{4\sqrt{3}}\begin{pmatrix}
1+\sqrt{5} & 1-i \\ 1+i & -1+\sqrt{5}
\end{pmatrix},\qquad
R_2=\frac{1}{12\sqrt{3}}\begin{pmatrix}
1+3\sqrt{5} & -5-i \\ -5+i & -1+3\sqrt{5}
\end{pmatrix},\\
R_3&=\frac{1}{12\sqrt{3}}\begin{pmatrix}
-1-3\sqrt{5} & -1-5i \\ -1+5i & 1-3\sqrt{5}
\end{pmatrix},\qquad
R_4=\frac{1}{12\sqrt{3}}\begin{pmatrix}
5-3\sqrt{5} & -1+i \\ -1-i & -5-3\sqrt{5}
\end{pmatrix}.
\end{split}
\end{equation}
Observe that $R_1$ and $R_2$ have positive eigenvalues, while the eigenvalues of $R_3$ and $R_4$ are negative. The traces of $R_k$ are real and have the same module,
\begin{equation}
\Tr(R_1)=\Tr(R_2)=\frac{\sqrt{5}}{2\sqrt{3}},\qquad
\Tr(R_3)=\Tr(R_4)=-\frac{\sqrt{5}}{2\sqrt{3}}.
\end{equation}
Now, it is straightforward to show that
\begin{equation}\label{ex1}
\sum_{k=1}^4R_k\otimes R_k=\frac 16
\left[\begin{array}{c c|c c}
3 & 0 & 0 & 0 \\
0 & 2 & 1 & 0 \\
\hline
0 & 1 & 2 & 0 \\
0 & 0 & 0 & 3
\end{array}\right],
\end{equation}
and hence $R_k$ satisfy eq. (\ref{conic}) with $\kappa_+=1/3$ and $\kappa_-=1/6$.
However, $R_k$ do not form a conical 2-design due to $R_k\ngeq 0$. This proves that semi-positivity is indeed an essential requirement in the definition of conical designs.
\end{Example}

Conical 2-designs are also formed by certain classes of positive, operator-valued measures (POVMs), where $R_k$ are semi-positive operators such that $\sum_{k=1}^MR_k=\mathbb{I}_d$. Examples include symmetric, informationally complete (SIC) POVMs \cite{Renes}, general SIC POVMs \cite{Gour,Yoshida}, mutually unbiased bases (MUBs) \cite{Schwinger,Szarek}, mutually unbiased measurements (MUMs) \cite{Kalev,Wang}, $(N,M)$-POVMs \cite{SIC-MUB}, generalized symmetric measurements \cite{SIC-MUB_general}, and generalized equiangular measurements \cite{GEAM}. A counterexample consists of semi-SIC POVMs \cite{semi-SIC} and equioverlapping measurements \cite{EOM22,EOM24,EOMq3}, which allow for operators of unequal traces.

In what follows, we prove that there exists a close relation between conical 2-designs and generalized equiangular tight frames.

\section{Generalized equiangular tight frames}

In quantum information theory, pure quantum states are represented by rank-1 projectors $|\psi_k\>\<\psi_k|$ onto the vectors $\psi_k$ in the Hilbert space $\mathcal{H}\simeq\mathbb{C}^d$. A more general construction is a line $\{P_k=a_k|\psi_k\>\<\psi_k|;\,a_k\in\mathbb{R}/\{0\}$, whose elements are rank-1 operators $P_k$ that are projective only for $a_k=1$. Of particular interest are collections $\{P_k;\,k=1,\ldots,M\}$ of lines that are uniform and equiangular; that is, $a_k=a$ and $\Tr(P_kP_\ell)=c\Tr(P_k)\Tr(P_\ell)$ for all $k\neq\ell$. Uniform equiangular lines for which $\sum_{k=1}^MP_k=\gamma\mathbb{I}_d$, $\gamma>0$, define an equiangular tight frame \cite{Strohmer2,Strohmer}. It has been shown that the number $M$ of frame elements $P_k$ belongs to the range $d\leq M\leq d^2$, where $M=d^2$ corresponds to a SIC POVM rescaled by the factor $\gamma$ \cite{Lemmens}. Recently, a generalization of equiangular tight frames has been considered, where the rank requirement of $P_k$ is dropped.

\begin{Definition}(\cite{GEAM})\label{GETF}
A generalized equiangular tight frame (GETF) $\mathcal{P}=\{P_k;\,k=1,\ldots,M\}$ is a collection of $2\leq M\leq d^2$ semi-positive, non-zero operators $P_k$ such that
\begin{equation}
\begin{split}
\Tr(P_k)&=a,\\
\Tr(P_k^2)&=b[\Tr(P_k)]^2,\\
\Tr(P_kP_\ell)&=c\Tr(P_k)\Tr(P_\ell),\qquad k\neq\ell,
\end{split}
\end{equation}
and $\sum_{k=1}^MP_k=\gamma\mathbb{I}_d$ with $\gamma>0$.
\end{Definition}

Consequently, this definition fixes the values of both the trace and the overlap,
\begin{equation}\label{ac}
a=\frac{d\gamma}{M},\qquad c=\frac{M-db}{d(M-1)},
\end{equation}
as well as the admissible range of the free parameter $b$,
\begin{equation}\label{brange}
\frac{1}{d}<b\leq\frac 1d \min\{d,M\}.
\end{equation}
We exclude $b=1/d$ that corresponds to the trivial case with $P_k=\gamma\mathbb{I}_d/M$. The choice $b=1$ corresponds to rank-1 projectors. If $b=M/d$, then $P_k$ are projectors or rank $d/M$.
Additionally, all the elements of a GETF are always linearly independent. For a detailed derivation of these properties, see Appendix A. Note that for $\gamma=1$, GETFs are quantum measurements known as $(1,M)$-POVMs \cite{SIC-MUB}, and taking $a=1$ reproduces (projective) equiangular measurements \cite{EOM22}.

Following the methods presented in refs. \cite{Gour,SIC-MUB_general}, we show that any GETF can be constructed from a set $\mathcal{G}=\{G_k;\,k=0,\ldots,M-1\}$ of orthonormal Hermitian operators such that $G_0=\mathbb{I}_d/\sqrt{d}$ and $\Tr(G_k)=0$ for all $k=1,\ldots,M-1$. In particular, for maximal GETFs ($M=d^2$), $\mathcal{G}$ is an operator basis. Now, using $G_k$, let us define $M$ traceless operators
\begin{equation}\label{H}
H_k=\left\{\begin{aligned}
&G-\sqrt{M}(1+\sqrt{M})G_k,\quad k=1,\ldots,M-1,\\
&(1+\sqrt{M})G,\qquad k=M,
\end{aligned}\right.
\end{equation}
where $G=\sum_{k=1}^{M-1}G_k$. Then, the corresponding GETF is given by
\begin{equation}
P_k=\frac{\gamma}{M}\mathbb{I}_d+\tau H_k,
\end{equation}
provided that the real parameter $\tau$ is chosen so that it guarantees the semi-positivity of every $P_k$. Observe that $\tau$ is related to the GETF parameters $(a,b,c)$ via
\begin{equation}
\tau=\pm\sqrt{\frac{a^2(b-c)}{M(\sqrt{M}+1)^2}}.
\end{equation}
Another way to construct a GETF from the same Hermitian orthonormal basis is to replace $H_k$ with
\begin{equation}\label{Hprime}
H_k^\prime=\left\{\begin{aligned}
&G+\sqrt{M}(1-\sqrt{M})G_k,\quad k=1,\ldots,M-1,\\
&(1-\sqrt{M})G,\qquad k=M.
\end{aligned}\right.
\end{equation}
Then, one arrives at an alternative formula
\begin{equation}
P_k^\prime=\frac{\gamma}{M}\mathbb{I}_d+\tau^\prime H_k^\prime
\end{equation}
for a GETF $\mathcal{P}^\prime=\{P_k;\,k=1,\ldots,M\}$ with
\begin{equation}
\tau^\prime=\pm\sqrt{\frac{a^2(b-c)}{M(\sqrt{M}-1)^2}}.
\end{equation}
Note that a single triple $(a,b,c)$ is associated with four numbers: $\pm|\tau|$, $\pm|\tau^\prime|$. Moreover, in general, $\mathcal{P}^\prime\neq \mathcal{P}$, and hence the same $\mathcal{G}$ can be used to introduce more than one GETF. For more details, see the results in ref. \cite{SIC-MUB_general} for $N=1$.

\section{Conical 2-designs vs. generalized equiangular tight frames}

In this section, we establish an equivalence relation between conical 2-designs and GETFs. Note that this is possible only for sets of $M=d^2$ operators. We divide this task into proving two implication relations.

\begin{Proposition}\label{GETFtoconical}
If $\mathcal{P}=\{P_k;\,k=1,\ldots,d^2\}$ is a generalized equiangular tight frame, then it is a conical 2-design.
\end{Proposition}

\begin{proof}
Let us start by constructing the map $\Phi[X]=\sum_{k=1}^{d^2}P_k\Tr(P_kX)$, multiplying it by $P_\ell$, and taking the trace. This results in
\begin{equation}\label{trace}
\begin{split}
\Tr(\Phi[X]P_\ell)&=\sum_{k=1}^{d^2}\Tr(P_kP_\ell)\Tr(P_kX)=\Tr(P_\ell^2)\Tr(P_\ell X)+\sum_{k\neq\ell}\Tr(P_kP_\ell)\Tr(P_kX)\\&=a^2b\Tr(P_\ell X)+a^2c\left[\sum_{k=1}^{d^2}\Tr(P_kX)-\Tr(P_\ell X)\right]
=a^2(b-c)\Tr(P_\ell X)+a^2c\gamma\Tr(X),
\end{split}
\end{equation}
where we have used the properties of $P_k$ from Definition \ref{GETF}. 
Equivalently, this formula can be rewritten into
\begin{equation}
\Tr\left\{\left[-\Phi+a^2(b-c)\oper+acd\gamma\Phi_0\right][X]P_\ell\right\}=0,
\end{equation}
and it has to hold for every $X$ and $P_\ell$. As $\mathcal{P}$ spans the operator space on $\mathcal{H}$, this condition reduces to
\begin{equation}
\Phi=a^2(b-c)\oper+acd\gamma\Phi_0.
\end{equation}
Comparing the above equation with eq. (\ref{kap}), we find that
\begin{equation}
\kappa_+=ac\gamma,\qquad \kappa_-=a^2(b-c).
\end{equation}
Now, eqs. (\ref{ac}) and (\ref{brange}) allow us to check that indeed
\begin{equation}
\kappa_+=\gamma^2\frac{d-b}{d(d^2-1)}\geq\kappa_-=\gamma^2\frac{db-1}{d(d^2-1)}>0
\end{equation}
for every $b\leq 1$.
\end{proof}

Therefore, all generalized equiangular tight frame with $M=d^2$ elements are (homogeneous) conical 2-designs. Next, we prove an inverse implication relation for conical 2-designs that are Hermitian operator bases.

\begin{Proposition}\label{conicalgeneral}
If $\mathcal{R}=\{R_k;\,k=1,\ldots,d^2\}$ with linearly independent operators $R_k$ is a conical 2-design, then its elements obey the trace relations
\begin{equation}
\begin{split}
\Tr(R_k)&=w_k,\\
\Tr(R_k^2)&=\kappa_-+\frac{\kappa_+}{\kappa}w_k^2,\\
\Tr(R_kR_\ell)&=\frac{\kappa_+}{\kappa}w_kw_\ell,\qquad k\neq\ell,
\end{split}
\end{equation}
where $\sum_{k=1}^{d^2}w_k^2=d\kappa$ with $\kappa=d\kappa_++\kappa_-$.
\end{Proposition}

\begin{proof}
Assume that a conical 2-design satisfies the most general trace relations:
\begin{equation}
\begin{split}
\Tr(R_k)&=w_k,\\
\Tr(R_k^2)&=x_k,\\
\Tr(R_kR_\ell)&=y_{k\ell},\qquad k\neq\ell,
\end{split}
\end{equation}
where $w_k$, $x_k$, and $y_{k\ell}=y_{\ell k}$ are real numbers due to the Hermiticity of $R_k$. From eq. (\ref{conicsum}), one also has
\begin{equation}\label{sum}
\sum_{k=1}^{d^2}w_kR_k=\kappa\mathbb{I}_d.
\end{equation}
Using the trace relations, we find
\begin{equation}\label{cc}
\sum_{k=1}^{d^2}w_k^2=\sum_{k=1}^{d^2}w_k\Tr(R_k)=\kappa\Tr(\mathbb{I}_d)=d\kappa,
\qquad
\kappa w_k=\kappa\Tr(R_k)=\sum_{\ell=1}^{d^2}w_\ell\Tr(R_kR_\ell)
=w_kx_k+\sum_{\ell\neq k}w_\ell y_{k\ell}.
\end{equation}

Recall that $\mathcal{R}$ is a conical 2-design if and only if
\begin{equation}\label{kap2}
\Phi=\kappa_+d\Phi_0+\kappa_-\oper,\qquad \kappa_+\geq\kappa_->0,
\end{equation}
where $\Phi[X]=\sum_{\ell=1}^{d^2}R_\ell\Tr(R_\ell X)$. For the left hand-side of eq. (\ref{kap2}), we calculate
\begin{equation}\label{LHS}
\begin{split}
\Tr(\Phi[X]R_k)&=\sum_{\ell=1}^{d^2}\Tr(R_\ell R_k)\Tr(R_\ell X) 
=\Tr(R_k^2)\Tr(R_k X)+\sum_{\ell\neq k}\Tr(R_\ell R_k)\Tr(R_\ell X)\\&
=x_k\Tr(R_k X)+\sum_{\ell\neq k}y_{k\ell}\Tr(R_\ell X),
\end{split}
\end{equation}
and, for the right hand-side,
\begin{equation}\label{RHS}
\Tr(\Phi[X]R_k)=\kappa_+w_k\Tr(X)+\kappa_-\Tr(R_kX).
\end{equation}
By comparing eqs. (\ref{LHS}) and (\ref{RHS}), we obtain the formula
\begin{equation}
\Tr\left\{\left[(x_k-\kappa_-)R_k+\sum_{\ell\neq k}y_{k\ell}R_\ell
-\kappa_+w_k\mathbb{I}_d\right]X\right\}=0
\end{equation}
that holds for all Hermitian $X$. Hence,
\begin{equation}\label{bb}
(x_k-\kappa_-)R_k+\sum_{\ell\neq k}y_{k\ell}R_\ell
=w_k\kappa_+\mathbb{I}_d,
\end{equation}
or, in terms of $R_k$ only,
\begin{equation}
[\kappa(x_k-\kappa_-)-\kappa_+w_k^2]R_k+\sum_{\ell\neq k}
(\kappa y_{k\ell}-w_kw_\ell\kappa_+)R_\ell=0.
\end{equation}
But all $R_k$ are linearly independent, so every term multiplying $R_k$ vanishes. Therefore,
\begin{equation}\label{xy}
x_k=\kappa_-+\frac{\kappa_+}{\kappa}w_k^2,\qquad 
y_{k\ell}=\frac{\kappa_+}{\kappa}w_kw_\ell.
\end{equation}
Finally, we find that the Cauchy-Schwarz inequality does not introduce any new conditions due to the inequalities
\begin{equation}
\begin{split}
w_k^2&=[\Tr(R_k)]^2\leq\Tr(R_k^2)\Tr(\mathbb{I}_d)=dx_k,\\
y_{k\ell}^2&=[\Tr(R_kR_\ell)]^2\leq\Tr(R_k^2)\Tr(R_\ell^2)=x_kx_\ell
\end{split}
\end{equation}
being always satisfied for $x_k$ and $y_{k\ell}$ given by eq. (\ref{xy}).
\end{proof}

\begin{Example}
As an illustrative example of a conical 2-design with elements of unequal traces, consider the following semi-positive qubit operators,
\begin{equation}
\begin{split}
R_1&=\frac{\sqrt{5}}{3\sqrt{2}}\begin{pmatrix}
1 & 0 \\ 0 & 0
\end{pmatrix},\qquad
R_2=\frac{1}{6\sqrt{5}}\begin{pmatrix}
4 & 3 \\ 3 & 6
\end{pmatrix},\\
R_3&=\frac{1}{3\sqrt{10}}\begin{pmatrix}
2 & -1-i\sqrt{5} \\ -1+i\sqrt{5} & 3
\end{pmatrix},\qquad
R_4=\frac{1}{6\sqrt{5}}\begin{pmatrix}
4 & -2+i\sqrt{5} \\ -2-i\sqrt{5} & 6
\end{pmatrix}.
\end{split}
\end{equation}
Observe that $\Tr(R_1)=\Tr(R_3)=\sqrt{5}/(3\sqrt{2})$ and $\Tr(R_2)=\Tr(R_4)=\sqrt{5}/3$. Moreover, $R_1$ and $R_3$ are rank-1 operators, whereas $R_2$ and $R_4$ are full rank operators with the same spectrum:
\begin{equation}
\left\{\frac{\sqrt{10}+5}{6\sqrt{5}},\,\frac{\sqrt{10}-5}{6\sqrt{5}}\right\}.
\end{equation}
Now, we calculate $\sum_{k=1}^4R_k\otimes R_k$, which gives results identical to eq. (\ref{ex1}). This shows that $R_k$ form a conical 2-design with $\kappa_+=1/3$ and $\kappa_-=1/6$.
\end{Example}

In general, conical 2-designs define a wider, less symmetric class than generalized equiangular tight frames. Regardless, there exists a close relation between homogeneous designs and GETFs. To show this, we first present the following properties.

\begin{Proposition}\label{equiv}
For a conical 2-design $\mathcal{R}=\{R_k;\,k=1,\ldots,d^2\}$ with linearly independent elements $R_k$, the following statements are equivalent:
\begin{enumerate}[label=(\it\roman*)]
\item $R_k$ are of equal trace: $\Tr(R_k)=w$;
\item $R_k^2$ are of equal trace: $\Tr(R_k^2)=x$;
\item $R_k$ satisfy the symmetry condition $\Tr(R_kR_\ell)=y$ for all $k\neq\ell$;
\item $R_k$ sum up to a rescaled identity operator $\sum_{k=1}^{d^2}R_k=\eta\mathbb{I}_d$ for some $\eta>0$;
\item $\mathcal{R}$ is a homogeneous conical 2-design.
\end{enumerate}
\end{Proposition}

\begin{proof}
Let us use the notation from Proposition \ref{conicalgeneral}. Obviously,
\begin{equation}
(i)\quad\bigforall_k\quad w_k=w\quad\Longleftrightarrow\quad
(iii)\quad \bigforall_{k\neq\ell}\quad y_{k\ell}=\frac{\kappa_+}{\kappa}w_kw_\ell=y,
\end{equation}
and also
\begin{equation}
(i)\quad\bigforall_k\quad w_k=w\quad\Longleftrightarrow\quad
(ii)\quad \bigforall_{k}\quad x_k=\kappa_-+\frac{\kappa_+}{\kappa}w_k^2=x,
\end{equation}
which in turn implies that $\mathcal{R}$ is a homogeneous conical 2-design. The inverse implication relation follows directly from the definition of homogeneous designs, and hence $(i)$ $\Longleftrightarrow$ $(v)$. Next, note that if $w_k=w$, then one immediately has $(iv)$ with $\eta=\kappa/w$ due to $\sum_{k=1}^{d^2}w_kR_k=w\sum_{k=1}^{d^2}R_k=\kappa\mathbb{I}_d$. Conversely, if $(iv)$ holds, then
\begin{equation}
\kappa\sum_{k=1}^{d^2}R_k=\kappa\eta\mathbb{I}_d=\eta\sum_{k=1}^{d^2}w_kR_k,
\end{equation}
where we used eq. (\ref{sum}) in the second equality. Hence, we have
\begin{equation}
\sum_{k=1}^{d^2}(\kappa-\eta w_k)R_k=0,
\end{equation}
which recovers $w_k=\kappa/\eta=w$ from linear independence of $R_k$, and therefore $(i)$ $\Longleftrightarrow$ $(iv)$.
\end{proof}

Finally, a one-to-one correspondence can be established between homogeneous conical 2-designs and GETFs.

\begin{Proposition}\label{conicaltoGETF}
A homogeneous conical 2-design $\mathcal{R}=\{R_k;\,k=1,\ldots,d^2\}$ with linearly independent elements $R_k$ is a GETF.
\end{Proposition}

\begin{proof}
Consider a homogeneous conical 2-design, where by definition
\begin{equation}
\begin{split}
\Tr(R_k)&=w=\sqrt{\frac{\kappa}{d}},\\
\Tr(R_k^2)&=x=\kappa_-+\frac{\kappa_+}{d},\\
\Tr(R_kR_\ell)&=y=\frac{\kappa_+}{d},\qquad k\neq\ell,
\end{split}
\end{equation}
and $\sum_{k=1}^{d^2}R_k=\kappa\mathbb{I}_d/w$.
Now, comparing these coefficients with Definition \ref{GETF} of generalized equiangular tight frames, we identify $\gamma=\kappa/w>0$, as well as
\begin{equation}
a=w,\qquad b=\frac{x}{w^2},\qquad c=\frac{y}{w^2}.
\end{equation}
It remains to check whether the above choices of $a$, $b$, and $c$ indeed satisfy all the properties of a GETF. Obviously, $a>0$ and $c>0$. The range of $b$ follows from the conditions on $\kappa_\pm$. Observe that if $\kappa_->0$, then
\begin{equation}
b=\frac{x}{w^2}=\frac 1d \left(1+\frac{d^2-1}{\kappa}\kappa_-\right)>\frac 1d.
\end{equation}
Finally, if $\kappa_+\geq\kappa_-$, then $b$ is upper bounded by
\begin{equation}
b_{\max}=b\Big|_{\kappa_+=\kappa_-}=1.
\end{equation}
Therefore, $b$ corresponding to a homogeneous conical 2-design lies in the admissible range from eq. (\ref{brange}), which guarantees $R_k\geq 0$.
\end{proof}

By combining the results from Propositions \ref{GETFtoconical} and \ref{conicaltoGETF}, we formulate our first main result.

\begin{Theorem}\label{main1}
A collection $\mathcal{R}=\{R_k;\,k=1,\ldots,d^2\}$ of linearly independent operators $R_k$ is a homogeneous conical 2-design if and only if it is a generalized equiangular tight frame.
\end{Theorem}

Recall that a GETF is a quantum measurements (POVM) if $\gamma=1$. More specifically, it is a general SIC POVM \cite{Gour}. By definition, the POVM elements $P_k$ sum up to the identity and hence satisfy all the statements in Proposition \ref{equiv}. This allows us to formulate another relation.

\begin{Theorem}\label{conicalPOVM}
Consider a POVM $\mathcal{P}=\{P_k;\,k=1,\ldots,d^2\}$ whose elements $P_k$ are linearly independent. Then, $\mathcal{P}$ is a conical 2-design if and only if it is a general SIC POVM.
\end{Theorem}

The scope of our considerations is limited to maximal sets of semi-positive operators, which form informationally complete sets. However, the results of this section can easily be generalized to informationally overcomplete. To show this, we introduce families of GETFs that satisfy additional symmetry constraints.

\section{Mutually unbiased generalized equiangular tight frames}

Consider a collection of $N$ generalized equiangular tight frames $\mathcal{P}_\alpha=\{P_{\alpha,k};\,k=1,\ldots,M_\alpha\}$, where $\sum_{k=1}^{M_\alpha}P_{\alpha,k}=\gamma_\alpha\mathbb{I}_d$ with $\gamma_\alpha>0$. Each frame consists of $M_\alpha\geq 2$ semi-positive operators, and there are no constraints on their ranks. Assume that these GETFs are complementary \cite{PetzRuppert} (or mutually unbiased \cite{Fickus,Goyeneche2}), which means that their elements satisfy $\Tr(P_{\alpha,k}P_{\beta,\ell})=f\Tr(P_{\alpha,k})\Tr(P_{\beta,\ell})$ for any $\alpha\neq\beta$. This allows us to introduce the following definition.

\begin{Definition}\label{MUGETFS}
Mutually unbiased generalized equiangular tight frames (MU GETFs) $\mathcal{P}=\cup_{\alpha=1}^N\mathcal{P}_\alpha$ are a collection of $N$ generalized equiangular tight frames $\mathcal{P}_\alpha=\{P_{\alpha,k};\,k=1,\ldots,M_\alpha\}$, $\alpha=1,\ldots,N$, such that
\begin{equation}
\begin{split}
\Tr(P_{\alpha,k})&=a_\alpha,\\
\Tr(P_{\alpha,k}^2)&=b_\alpha \Tr(P_{\alpha,k})^2,\\
\Tr(P_{\alpha,k}P_{\alpha,\ell})&=c_\alpha
\Tr(P_{\alpha,k})\Tr(P_{\alpha,\ell}),\qquad k\neq\ell,\\
\Tr(P_{\alpha,k}P_{\beta,\ell})&=
f\Tr(P_{\alpha,k})\Tr(P_{\beta,\ell}),\qquad \alpha\neq\beta,
\end{split}
\end{equation}
and $\sum_{k=1}^{M_\alpha}P_{\alpha,k}=\gamma_\alpha\mathbb{I}_d$ for $\gamma_\alpha>0$.
\end{Definition}

Using the results in eqs. (\ref{ac}) and (\ref{brange}), we immediately recover the coefficients
\begin{equation}\label{ac2}
a_\alpha=\frac{d\gamma_\alpha}{M_\alpha},\qquad 
c_\alpha=\frac{M_\alpha-db_\alpha}{d(M_\alpha-1)},
\end{equation}
and the ranges of the free parameters $b_\alpha$,
\begin{equation}\label{brange2}
\frac{1}{d}<b_\alpha\leq\frac 1d \min\{d,M_\alpha\}.
\end{equation}
Lastly, the value of $f=1/d$ follows from
\begin{equation}
d\gamma_\alpha\gamma_\beta=\sum_{k=1}^{M_\alpha}\sum_{\ell=1}^{M_\beta}
\Tr(P_{\alpha,k}P_{\beta,\ell})=f\sum_{k=1}^{M_\alpha}\sum_{\ell=1}^{M_\beta}
\Tr(P_{\alpha,k})\Tr(P_{\beta,\ell})=d^2f\gamma_\alpha\gamma_\beta,
\end{equation}
which holds for every $\alpha\neq\beta$.

Note that, in the most general case, $\mathcal{P}$ itself is not a GETF because it does not always satisfy the trace relations in Definition \ref{GETF}, even though $\sum_{\alpha=1}^N\sum_{k=1}^{M_\alpha}P_{\alpha,k}=\Gamma\mathbb{I}_d$ with $\Gamma=\sum_{\alpha=1}^N\gamma_\alpha$. The special case where MU GETFs form a GETF corresponds to the choice $\gamma_\alpha=aM_\alpha/d$ and $b_\alpha=b$, where $a$ is a positive constant and $1/d<b\leq\min_\alpha\{d,M_\alpha\}$.

In Appendix B, we prove that the set
$\{\mathbb{I}_d,\,P_{\alpha,k};\,k=1,\ldots,M_\alpha-1;\,\alpha=1,\ldots,N\}$
consists of linearly independent operators, and the total number $|\mathcal{P}|$ of elements of $\mathcal{P}$ is bounded by
\begin{equation}\label{bounds}
2N\leq\sum_{\alpha=1}^NM_\alpha\leq d^2+N-1.
\end{equation}
If and only if the upper bound is reached, $P_{\alpha,k}\in\mathcal{P}$ form an informationally overcomplete set. Additionally taking $\Gamma=1$ results in generalized equiangular measurements \cite{GEAM}.

Similarly to Section 3, we construct MU GETFs from a set of orthonormal Hermitian operators $\mathcal{G}=\{G_0,\,G_{\alpha,k};\,k=1,\ldots,M_\alpha-1;\,\alpha=1,\ldots,N\}$ that consists of $G_0=\mathbb{I}_d/\sqrt{d}$ and traceless $G_{\alpha,k}$. Note that this time the elements of $\mathcal{G}$ are assigned two indices, and the resulting GETFs strongly depend on the partition of $\mathcal{G}$ into $\mathcal{G}_\alpha=\{G_{\alpha,k};\,k=1,\ldots,M_\alpha-1\}$. This can be seen from the construction method alone, as the mutually unbiased GETFs are given by
\begin{equation}
P_{\alpha,k}=\frac{\gamma_\alpha}{M_\alpha}\mathbb{I}_d+\tau_\alpha H_{\alpha,k},
\end{equation}
where
\begin{equation}\label{Hak}
H_{\alpha,k}=\left\{\begin{aligned}
&G_\alpha-\sqrt{M_\alpha}(1+\sqrt{M_\alpha})G_{\alpha,k},\quad k=1,\ldots,M_\alpha-1,\\
&(1+\sqrt{M_\alpha})G_\alpha,\qquad k=M_\alpha,
\end{aligned}\right.
\end{equation}
$G_\alpha=\sum_{k=1}^{M_\alpha-1}G_{\alpha,k}$, and $\tau_\alpha$ is chosen to guarantee that $P_{\alpha,k}\geq 0$ for all $k=1,\ldots,M_\alpha$ and a fixed $\alpha$. The relation between $\tau_\alpha$ and the GETF parameters $(a_\alpha,b_\alpha,c_\alpha)$ reads
\begin{equation}
\tau_\alpha=\pm\sqrt{\frac{S_\alpha}{M_\alpha(\sqrt{M_\alpha}+1)^2}},\qquad S_\alpha\equiv a_\alpha^2(b_\alpha-c_\alpha).
\end{equation}
Alternatively, any number of GETFs $\mathcal{P}_\alpha^\prime=\{P_{\alpha,k}^\prime;\,k=1,\ldots,M_\alpha\}$ can be constructed through the formula
\begin{equation}
P_{\alpha,k}^\prime=\frac{\gamma_\alpha}{M_\alpha}\mathbb{I}_d+\tau_\alpha^\prime H_{\alpha,k}^\prime,
\end{equation}
where the choice of $\tau_\alpha^\prime$ guarantees that $P_{\alpha,k}^\prime\geq 0$ and
\begin{equation}\label{Hak2}
H_{\alpha,k}^\prime=\left\{\begin{aligned}
&G_\alpha+\sqrt{M_\alpha}(1-\sqrt{M_\alpha})G_{\alpha,k},\quad k=1,\ldots,M_\alpha-1,\\
&(1-\sqrt{M_\alpha})G_\alpha,\qquad k=M_\alpha.
\end{aligned}\right.
\end{equation}
Once again, there are two possible choices of $\tau_\alpha^\prime$ for the same triple $(a_\alpha,b_\alpha,c_\alpha)$,
\begin{equation}
\tau_\alpha^\prime=\pm\sqrt{\frac{S_\alpha}{M_\alpha(\sqrt{M_\alpha}-1)^2}},\qquad S_\alpha\equiv a_\alpha^2(b_\alpha-c_\alpha).
\end{equation}
In general, $\mathcal{P}_\alpha^\prime\neq \mathcal{P}_\alpha$ \cite{SIC-MUB_general}, and hence the same partition of $\mathcal{G}$ into $\mathcal{G}_\alpha$ results in up to four one-parameter families of GETFs. Maximal sets of MU GETFs are informationally overcomplete, with $N-1$ informationally redundant elements, and they are constructed from orthonormal bases $\mathcal{G}$. Finally, observe that $\mathcal{P}=\cup_{\alpha=1}^N\widetilde{\mathcal{P}}_\alpha$, where $\widetilde{\mathcal{P}}_\alpha\in\{\mathcal{P}_\alpha,\mathcal{P}^\prime_\alpha\}$. Therefore, one set of orthonormal Hermitian operators $\mathcal{G}$ can lead to many unequivalent constructions of MU GETFs.

Of particular interest is a class of informationally overcomplete MU GETFs for which there exists a direct relation between the index of coincidence \cite{Rastegin5}
\begin{equation}\label{pak}
\mathcal{C}(\rho)=\sum_{\alpha=1}^N\sum_{k=1}^{M_\alpha}p_{\alpha,k}^2(\rho),\qquad p_{\alpha,k}(\rho)=\frac{1}{\Gamma}\Tr(P_{\alpha,k}\rho),
\end{equation}
and the state purity $\Tr(\rho^2)$. The rescaling by $1/\Gamma$ is necessary for $p_{\alpha,k}$ to form a probability distribution. Observe that the operators $P_{\alpha,k}/\Gamma$ form a GEAM, and therefore the corresponding index of coincidence reads \cite{GEAM}
\begin{equation}\label{IOCN}
\mathcal{C}(\rho)=S\left(\Tr\rho^2-\frac 1d\right)+\mu,
\end{equation}
where
\begin{equation}\label{mu}
\mu=\frac 1d \sum_{\alpha=1}^Na_\alpha \gamma_\alpha,
\end{equation}
but only under the additional condition that $S_\alpha=a_\alpha(b_\alpha-c_\alpha)\equiv S$. As it turns out, this constraint is closely related to yet another property of MU GETFs. Indeed, let us calculate the Frobenius distance between the elements $P_{\alpha,k}$ of a single GETF $\mathcal{P}_\alpha$ \cite{Orlowski},
\begin{equation}
D_2^2(X,Y)=\frac 12\|X-Y\|_2^2,
\end{equation}
where $\|X\|_2=\sqrt{\Tr(X^\dagger X)}$ is the Frobenius norm of $X$. Now, for any $k\neq\ell$, we have
\begin{equation}
D_2^2(P_{\alpha,k},P_{\alpha,\ell})=\frac 12\Tr[(P_{\alpha,k}-P_{\alpha,\ell})^2]
=\frac 12 \Big[\Tr(P_{\alpha,k}^2)+\Tr(P_{\alpha,\ell}^2)-2\Tr(P_{\alpha,k}P_{\alpha,k})\Big]
=a_\alpha^2(b_\alpha-c_\alpha)\equiv S_\alpha.
\end{equation}

\begin{Definition}\label{equidistant}
Mutually unbiased generalized equiangular tight frames are equidistant if $S_\alpha=S$ for all $\alpha=1,\ldots,N$, where
\begin{equation}\label{Srange}
0<S\leq \min_\alpha\left\{\frac{d\gamma_\alpha^2}{M_\alpha},\frac{d-1}{M_\alpha-1}
\frac{d\gamma_\alpha^2}{M_\alpha}\right\},
\end{equation}
\end{Definition}

A detailed derivation of the range of $S$ is presented in Appendix C.

\section{Conical 2-designs vs. mutually unbiased GETFs}

This section generalizes the results of Section 4 to conical 2-designs whose elements are linearly dependent while still spanning the space of operators on $\mathcal{H}\simeq\mathbb{C}^d$. In analogy to mutually unbiased (MU) GETFs, conical 2-designs are now numbered via two sets of indices, so that $\mathcal{R}=\{R_{\alpha,k};\,k=1,\ldots,M_\alpha;\,\alpha=1,\ldots,N\}$. The total number of operators $R_{\alpha,k}$ is given by the upper bound from eq. (\ref{bounds}), which guarantees that $\mathcal{R}$ is an informationally overcomplete (IOC) set. From now on, we assume that all sets of MU GETFs are maximal, which means that there is the total of $|\mathcal{P}|=\sum_{\alpha=1}^NM_\alpha=d^2+N-1$ elements $P_{\alpha,k}$.

Again, we start from proving an implication relation for the GETFs.

\begin{Proposition}\label{MUGETFStoconical}
If $\mathcal{P}_\alpha=\{P_{\alpha,k};\,k=1,\ldots,M_\alpha\}$, $\alpha=1,\ldots,N$, are equidistant mutually unbiased generalized equiangular tight frames, then $\mathcal{P}=\cup_{\alpha=1}^N\mathcal{P}_\alpha$ is a conical 2-design.
\end{Proposition}

\begin{proof}
The proof is analogical to that of Proposition \ref{GETFtoconical}.
For every GETF $\mathcal{P}_\alpha$, let us construct a map $\Phi_\alpha[X]=\sum_{k=1}^{M_\alpha}P_{\alpha,k}\Tr(P_{\alpha,k}X)$. From eq. (\ref{trace}), we immediately get
\begin{equation}
\begin{split}
\Tr(\Phi_\alpha[X]P_{\alpha,\ell})=a_\alpha^2(b_\alpha-c_\alpha)\Tr(P_{\alpha,\ell} X)+a_\alpha^2c_\alpha\gamma_\alpha\Tr(X).
\end{split}
\end{equation}
Now, for any $\beta\neq\alpha$, it follows from Definition \ref{MUGETFS} that
\begin{equation}
\begin{split}
\Tr(\Phi_\alpha[X]P_{\beta,\ell})&=\sum_{k=1}^{M_\alpha}\Tr(P_{\alpha,k}P_{\beta,\ell})
\Tr(P_{\alpha,k}X)=fa_\beta\sum_{k=1}^{M_\alpha}a_\alpha \Tr(P_{\alpha,k}X)
=f\gamma_\alpha a_\alpha a_\beta \Tr(X).
\end{split}
\end{equation}
Therefore, for the total map $\Phi=\sum_{\alpha=1}^N\Phi_\alpha$, one finds the formula
\begin{equation}\label{phi}
\begin{split}
\Tr(\Phi[X]P_{\alpha,k})&=a_\alpha^2(b_\alpha-c_\alpha)\Tr(P_{\alpha,k}X)
+\left[a_\alpha^2\gamma_\alpha(c_\alpha-f)+a_\alpha f\sum_{\beta=1}^Na_\beta \gamma_\beta\right]\Tr(X)\\
&=a_\alpha^2(b_\alpha-c_\alpha)\Tr(P_{\alpha,k}X)
-\frac{a_\alpha}{d}\left[a_\alpha^2(b_\alpha-c_\alpha)-\sum_{\beta=1}^Na_\beta \gamma_\beta\right]\Tr(X)
\end{split}
\end{equation}
that is the same as for generalized equiangular measurements in ref. \cite{GEAM}. In the last equality, the following identity between the coefficients has been used,
\begin{equation}\label{caf}
c_\alpha-f=-\frac{b_\alpha-c_\alpha}{M_\alpha}.
\end{equation}
Now, eq. (\ref{phi}) can be rewritten as
\begin{equation}
\Tr\left\{\left[-\Phi+S_\alpha\oper+\left(\mu-\frac{S_\alpha}{d}\right)d\Phi_0\right][X]P_{\alpha,k}\right\}=0,
\end{equation}
where $S_\alpha=a_\alpha^2(b_\alpha-c_\alpha)$ and $\mu$ is given by eq. (\ref{mu}).
The above trace relation holds for all $X$ and $P_{\alpha,k}$, and hence
\begin{equation}
\Phi=S_\alpha\oper+\left(\mu-\frac{S_\alpha}{d}\right)d\Phi_0.
\end{equation}
Comparing our results with $\Phi=\kappa_+d\Phi_0+\kappa_-\oper$ from eq. (\ref{kap}), we notice that MU GETFs are conical 2-designs if and only if they are equidistant: $S_\alpha\equiv S$ for all $\alpha=1,\ldots,N$. Under this condition, we identify
\begin{equation}
\kappa_+=\mu-\frac{S}{d},\qquad \kappa_-=S.
\end{equation}
Finally, it is easy to ckeck that $\kappa_->0$ and
\begin{equation}
\kappa_+-\kappa_-=\mu-\frac{d+1}{d}S\geq 0.
\end{equation}
\end{proof}

Note that the parameters $\kappa_\pm$ for the equidistant MU GETFs can be rewritten as
\begin{equation}
\kappa_+=\mathcal{C}_{\max}-S,\qquad \kappa_-=S,
\end{equation}
where 
\begin{equation}
\mathcal{C}_{\max}=\frac{d-1}{d}S+\mu
\end{equation}
is the upper bound for the index of coincidence $\mathcal{C}(\rho)$, reached for pure states ($\Tr(\rho^2)=1$). Therefore, $\kappa_\pm$ are fully characterized by the maximal index of coincidence and the Frobenius distance $S$.

The conical 2-designs from Proposition \ref{MUGETFStoconical} are in general no longer homogeneous for $N\geq 2$. Moreover, there exists a large class of MU GETFs that go beyond conical 2-designs. Neither of these was observed for a single GETF ($N=1$).

Now, we formulate an inverse implication for informationally overcomplete conical 2-designs. The main difference from the case in Proposition \ref{conicalgeneral} is that now there exist partial sums of $\Tr(R_k)R_k$ equal to a rescaled identity operator.

\begin{Proposition}\label{conicalgeneral2}
Consider $N$ sets $\mathcal{R}_\alpha=\{R_{\alpha,k};\,k=1,\ldots,M_\alpha\}$ of linearly independent semi-positive operators that satisfy $\sum_{k=1}^{M_\alpha}\Tr(R_{\alpha,k})R_{\alpha,k}=\kappa_\alpha\mathbb{I}_d$ for some $\kappa_\alpha>0$ and together span the operator space $\mathcal{B}(\mathcal{H})$. If $\mathcal{R}=\cup_{\alpha=1}^N\mathcal{R}_\alpha$ with $|\mathcal{R}|=d^2+N-1$ is a conical 2-design, then its elements obey the trace relations
\begin{equation}\label{traces}
\begin{split}
\Tr(R_{\alpha,k})&=w_{\alpha,k},\\
\Tr(R_{\alpha,k}^2)&=\kappa_-+\frac{\kappa_\alpha-\kappa_-}{d\kappa_\alpha}w_{\alpha,k}^2,\\
\Tr(R_{\alpha,k}R_{\alpha,\ell})&=\frac{\kappa_\alpha-\kappa_-}{d\kappa_\alpha}w_{\alpha,k}w_{\alpha,\ell},
\qquad k\neq\ell,\\
\Tr(R_{\alpha,k}R_{\beta,\ell})&=\frac 1d w_{\alpha,k}w_{\beta,\ell},\qquad \alpha\neq\beta.
\end{split}
\end{equation}
\end{Proposition}

\begin{proof}
Let us start by observing that if $\sum_{k=1}^{M_\alpha}\Tr(R_{\alpha,k})R_{\alpha,k}=\kappa_\alpha\mathbb{I}_d$ and $|\mathcal{R}|=d^2+N-1$, then the set $\{\mathbb{I}_d,\,R_{\alpha,k};\,k=1,\ldots,M_\alpha-1;\,\alpha=1,\ldots,N\}$
consists of operators that are linearly independent (for more details, see Appendix B). Next, we assume that a conical 2-design obeys the most general trace relations:
\begin{equation}
\begin{split}
\Tr(R_{\alpha,k})&=w_{\alpha,k},\\
\Tr(R_{\alpha,k}^2)&=x_{\alpha,k},\\
\Tr(R_{\alpha,k}R_{\alpha,\ell})&=y_{\alpha,k,\ell},\qquad k\neq\ell,\\
\Tr(R_{\alpha,k}R_{\beta,\ell})&=z_{\alpha,k;\beta,\ell},\qquad \alpha\neq\beta.
\end{split}
\end{equation}
As $R_{\alpha,k}$ are Hermitian, all the numbers are real. Moreover, from the cyclicity of the trace, $y_{\alpha,k,\ell}=y_{\alpha,\ell,k}$ and $z_{\alpha,k;\beta,\ell}=z_{\beta,\ell;\alpha,k}$. From eq. (\ref{conicsum}), it follows that
\begin{equation}\label{kappas}
\sum_{\alpha=1}^N\kappa_\alpha=\kappa,\qquad{\rm where}\qquad
\sum_{k=1}^{M_\alpha}w_{\alpha,k}R_{\alpha,k}=\kappa_\alpha\mathbb{I}_d
\end{equation}
and $\kappa=d\kappa_++\kappa_-$. Additionally, the trace relations imply
\begin{equation}\label{cc2}
\begin{split}
\sum_{k=1}^{M_\alpha}w_{\alpha,k}^2&=\sum_{k=1}^{M_\alpha}w_{\alpha,k}\Tr(R_{\alpha,k})
=\kappa_\alpha\Tr(\mathbb{I}_d)=d\kappa_\alpha,\\
\kappa_\alpha w_{\alpha,k}&=\kappa_\alpha\Tr(R_{\alpha,k})
=\sum_{\ell=1}^{M_\alpha}w_{\alpha,\ell}\Tr(R_{\alpha,k}R_{\alpha,\ell})
=w_{\alpha,k}x_{\alpha,k}+\sum_{\ell\neq k}w_{\alpha,\ell} y_{\alpha,k,\ell},
\end{split}
\end{equation}
which is a direct generalization of eq. (\ref{cc}). For $\alpha\neq\beta$, we obtain
\begin{equation}\label{cc3}
\begin{split}
\kappa_\beta w_{\alpha,k}&=\kappa_\beta\Tr(R_{\alpha,k})
=\sum_{\ell=1}^{M_\beta}w_{\beta,\ell}\Tr(R_{\alpha,k}R_{\beta,\ell})
=\sum_{\ell=1}^{M_\beta}w_{\beta,\ell}z_{\alpha,k;\beta,\ell}.
\end{split}
\end{equation}
This, together with the first line of eq. (\ref{cc2}), gives
\begin{equation}
d\kappa_\alpha\kappa_\beta=\sum_{k=1}^{M_\alpha}\sum_{\ell=1}^{M_\beta}
w_{\alpha,k}w_{\beta,\ell}z_{\alpha,k;\beta,\ell}.
\end{equation}

Now, for each $\mathcal{R}_\alpha$, we construct $\Phi_\alpha[X]=\sum_{k=1}^{M_\alpha}R_{\alpha,k}\Tr(R_{\alpha,k}X)$ and then calculate
\begin{equation}
\begin{split}
\Tr(\Phi_\alpha[X]R_{\alpha,k})=\sum_{\ell=1}^{M_\alpha}\Tr(R_{\alpha,\ell}R_{\alpha,k})
\Tr(R_{\alpha,\ell}X)=x_{\alpha,k}\Tr(R_{\alpha,k}X)+\sum_{\ell\neq k}
y_{\alpha,k,\ell}\Tr(R_{\alpha,\ell}X),
\end{split}
\end{equation}
as well as, for any $\beta\neq\alpha$,
\begin{equation}
\begin{split}
\Tr(\Phi_\beta[X]R_{\alpha,k})&=\sum_{\ell=1}^{M_\beta}\Tr(R_{\beta,\ell}R_{\alpha,k})
\Tr(R_{\beta,\ell}X)=\sum_{\ell=1}^{M_\beta}z_{\alpha,k;\beta,\ell}\Tr(R_{\beta,\ell}X).
\end{split}
\end{equation}
Therefore, for the total map $\Phi=\sum_{\alpha=1}^N\Phi_\alpha$, we find
\begin{equation}\label{phi2}
\begin{split}
\Tr(\Phi[X]R_{\alpha,k})&=x_{\alpha,k}\Tr(R_{\alpha,k}X)+\sum_{\ell\neq k}
y_{\alpha,k,\ell}\Tr(R_{\alpha,\ell}X)+\sum_{\beta\neq\alpha}\sum_{\ell=1}^{M_\beta}z_{\alpha,k;\beta,\ell}
\Tr(R_{\beta,\ell}X).
\end{split}
\end{equation}

Next, we have to compare our results with the general formula. Recall that the collection $\mathcal{R}=\cup_{\alpha=1}^N\mathcal{R}_\alpha$ forms a conical 2-design if and only if
\begin{equation}\label{kap3}
\Phi=\kappa_+d\Phi_0+\kappa_-\oper,\qquad \kappa_+\geq\kappa_->0.
\end{equation}
Using the above equation, one gets
\begin{equation}\label{phi3}
\Tr(\Phi[X]R_{\alpha,k})=w_{\alpha,k}\kappa_+\Tr(X)+\kappa_-\Tr(R_{\alpha,k}X).
\end{equation}
By comparing eqs. (\ref{phi2}) and (\ref{phi3}), we obtain an equivalent formula for conical 2-designs,
\begin{equation}\label{trrr}
\Tr\left\{\left[(x_{\alpha,k}-\kappa_-)R_{\alpha,k}
+\sum_{\ell\neq k}y_{\alpha,k,\ell}R_{\alpha,\ell}
+\sum_{\beta\neq\alpha}\sum_{\ell=1}^{M_\beta}z_{\alpha,k;\beta,\ell}
R_{\beta,\ell}-w_{\alpha,k}\kappa_+\mathbb{I}_d\right]X\right\}=0,
\end{equation}
which holds for all Hermitian $X$. Therefore, the expression in the square brackets vanishes. In Appendix D, we show that this is the case provided that
\begin{equation}
\begin{split}
x_{\alpha,k}=\kappa_-+\frac{\kappa_\alpha-\kappa_-}{d\kappa_\alpha}w_{\alpha,k}^2,\qquad
y_{\alpha,k,\ell}=\frac{\kappa_\alpha-\kappa_-}{d\kappa_\alpha}w_{\alpha,k}w_{\alpha,\ell},\qquad
z_{\alpha,k;\beta,\ell}=\frac 1d w_{\alpha,k}w_{\beta,\ell}.
\end{split}
\end{equation}
Finally, the condition $\kappa_\alpha\geq\kappa_-$ guarantees that $y_{\alpha,k,\ell}\geq 0$.
\end{proof}

Interestingly, $\kappa_-$ controls the parameters in the trace relations, whereas $\kappa_+$ only appears as the constraint for the sum over all $\kappa_\alpha$.
Observe that the semi-positivity of $R_{\alpha,k}$ imposes the following conditions on the parameters in Proposition \ref{conicalgeneral2} (see Appendix E),
\begin{equation}\label{kappaa}
\kappa_\alpha\geq\kappa_-,\qquad \frac{d\kappa_-\kappa_\alpha}{(d-1)\kappa_\alpha+\kappa_-}\leq w_{\alpha,k}^2<d\kappa_\alpha.
\end{equation}
Also, $\mathcal{R}_\alpha$ are equidistant if and only if $w_{\alpha,k}=w_{\alpha,\ell}$, which can be checked by computing
\begin{equation}
D_2^2(R_{\alpha,k},R_{\alpha,\ell})=\kappa_-+\frac{\kappa_\alpha-\kappa_-} {2d\kappa_\alpha}(w_{\alpha,k}-w_{\alpha,\ell})^2.
\end{equation}

\begin{Example}
A good example of a conical 2-design that belongs to the class from Proposition \ref{conicalgeneral2} is given by two collections of qubit operators:
\begin{equation}
\begin{split}
R_{1,1}&=\frac {1}{\sqrt{6}} \begin{pmatrix} 1 & 0 \\ 0 & 0 \end{pmatrix},\\
R_{1,2}&=\frac {1}{\sqrt{6}} \begin{pmatrix} 0 & 0 \\ 0 & 1 \end{pmatrix},
\end{split}\qquad\qquad\qquad
\begin{split}
R_{2,1}&=\frac {1}{2\sqrt{30}} 
\begin{pmatrix} 2\sqrt{3} & \sqrt{5}-i\sqrt{2} \\ \sqrt{5}+i\sqrt{2} & 2\sqrt{3} \end{pmatrix},\\
R_{2,2}&=\frac {1}{2\sqrt{30}} 
\begin{pmatrix} 2\sqrt{3} & -\sqrt{5}-i\sqrt{2} \\ -\sqrt{5}+i\sqrt{2} & 2\sqrt{3} \end{pmatrix},\\
R_{2,3}&=\frac {1}{2\sqrt{15}} 
\begin{pmatrix} 2\sqrt{2} & i\sqrt{3} \\ -i\sqrt{3} & 2\sqrt{2} \end{pmatrix}.
\end{split}
\end{equation}
$\mathcal{R}_1$ is a rescaled von Neumann measurement, whereas $\mathcal{R}_2$ is composed of three full rank operators with $\Tr(R_{2,1})=\Tr(R_{2,2})\neq\Tr(R_{2,3})$. Note that $R_{1,1}+R_{1,2}=\mathbb{I}_2/\sqrt{6}$ and $R_{2,1}\Tr(R_{2,1})+R_{2,2}\Tr(R_{2,2})+R_{2,3}\Tr(R_{2,3})=2\mathbb{I}_2/3$. Moreover, $R_{2,1}$ and $R_{2,2}$ have the same eigenvalues. 
The parameters that characterize this conical design are $\kappa_+=1/3$ and $\kappa_-=1/6$, which are identical to those in our previous two examples. It proves that the same values of $\kappa_\pm$ can correspond to both a GETF and MU GETFs.
\end{Example}

Under additional conditions, an implication relation between conical 2-designs and MU GETFs can be established. Note that it goes beyond the homogeneous designs because the traces of operators from different GETFs do not have to be the same. However, if the traces of $R_{\alpha,k}\in\mathcal{R}_\alpha$ are equal, some interesting properties can be recovered.

\begin{Proposition}\label{equiv2}
If $\mathcal{R}$ is a conical 2-design from Proposition \ref{conicalgeneral2}, then the following statements are equivalent:
\begin{enumerate}[label=(\it\roman*)]
\item $R_{\alpha,k}\in\mathcal{R}_\alpha$ are of equal trace: $\Tr(R_{\alpha,k})=w_\alpha$;
\item $R_{\alpha,k}^2\in\mathcal{R}_\alpha$ are of equal trace: $\Tr(R_{\alpha,k}^2)=x_\alpha$;
\item $R_{\alpha,k}\in\mathcal{R}_\alpha$ satisfy the symmetry condition $\Tr(R_{\alpha,k}R_{\alpha,\ell})=y_\alpha$ for all $k\neq\ell$;
\item $R_{\alpha,k}\in\mathcal{R}_\alpha$ sum up to a rescaled identity operator $\sum_{k=1}^{M_\alpha}R_{\alpha,k}=\eta_\alpha\mathbb{I}_d$ for some $\eta_\alpha>0$;
\item each $\mathcal{R}_\alpha$ is a GETF;
\item $\mathcal{R}_\alpha$ are mutually unbiased: $\Tr(R_{\alpha,k}R_{\beta,\ell})=f\Tr(R_{\alpha,k})\Tr(R_{\beta,\ell})$ for all $\alpha\neq\beta$.
\end{enumerate}
\end{Proposition}

\begin{proof}
Using the notation from Proposition \ref{conicalgeneral2}, we immediately show that
\begin{equation}
(i)\quad\bigforall_k\quad w_{\alpha,k}=w_\alpha\quad\Longleftrightarrow\quad
(ii)\quad \bigforall_{k}\quad x_{\alpha,k}=\kappa_-+\frac{\kappa_\alpha-\kappa_-}{d\kappa_\alpha}w_{\alpha,k}^2 =x_\alpha,
\end{equation}
\begin{equation}
(i)\quad\bigforall_k\quad w_{\alpha,k}=w_\alpha\quad\Longleftrightarrow\quad
(iii)\quad \bigforall_{k\neq\ell}\quad y_{\alpha,k,\ell}=\frac{\kappa_\alpha-\kappa_-}{d\kappa_\alpha}w_{\alpha,k} w_{\alpha,\ell}=y_\alpha,
\end{equation}
and
\begin{equation}
(i)\quad\bigforall_k\quad w_{\alpha,k}=w_\alpha\quad\Longleftrightarrow\quad
(vi)\quad \bigforall_{\beta\neq\alpha}\quad z_{\alpha,k;\beta\ell}=\frac{1}{d}w_{\alpha} w_{\beta}= \frac{1}{d} \Tr(R_{\alpha,k})\Tr(R_{\beta,\ell}),
\end{equation}
which proves mutual unbiasedness of $\mathcal{R}_\alpha$ with $f=1/d$.
Moreover, if $w_{\alpha,k}=w_\alpha$, then $(iv)$ holds for $\eta_\alpha=\kappa_\alpha/w_\alpha>0$, which can be seen from
\begin{equation}
\sum_{k=1}^{M_\alpha}w_{\alpha,k}R_{\alpha,k}=w_\alpha\sum_{k=1}^{M_\alpha}R_{\alpha,k} =\kappa_\alpha\mathbb{I}_d.
\end{equation}
On the other hand, if $(iv)$ is satisfied, then
\begin{equation}
\kappa_\alpha\sum_{k=1}^{M_\alpha}R_{\alpha,k} =\kappa_\alpha\eta_\alpha\mathbb{I}_d=\eta_\alpha\sum_{k=1}^{M_\alpha}w_{\alpha,k}R_{\alpha,k},
\end{equation}
which implies that
\begin{equation}
\sum_{k=1}^{M_\alpha}(\kappa_\alpha-\eta_\alpha w_{\alpha,k})R_{\alpha,k}=0.
\end{equation}
As $R_{\alpha,k}\in\mathcal{R}_\alpha$ are linearly independent, one recovers $w_{\alpha,k}=\kappa_\alpha/\eta_\alpha=w_\alpha$, and hence $(i)$ $\Longleftrightarrow$ $(iv)$. Finally, by definition, $\mathcal{R}_\alpha$ satisfies conditions $(i)$ to $(iv)$ if and only if $(v)$ $\mathcal{R}_\alpha$ is a GETF.
\end{proof}

Now, let us expand these results to further establish a relation between conical 2-designs and MU GETFs.

\begin{Proposition}\label{conicaltoMUGETFS} 
If $\mathcal{R}=\{R_{\alpha,k};\,k=1,\ldots,M_\alpha,\,\alpha=1,\ldots,N\}$ is a conical 2-design from Proposition \ref{conicalgeneral2} with $\Tr(R_{\alpha,k})=w_\alpha$, then it is a maximal set of equidistant MU GETFs.
\end{Proposition}

\begin{proof}
The condition that the trace of every $R_{\alpha,k}\in\mathcal{R}_\alpha$ depends only on the index $\alpha$ transforms eq. (\ref{traces}) into
\begin{equation}\label{traces2}
\begin{split}
\Tr(R_{\alpha,k})&=w_{\alpha},\\
\Tr(R_{\alpha,k}^2)&=x_\alpha=\frac{w_\alpha^2}{d}+\frac{M_\alpha-1}{M_\alpha} \kappa_-,\\
\Tr(R_{\alpha,k}R_{\alpha,\ell})&=y_\alpha=\frac{w_\alpha^2}{d}-\frac{\kappa_-}{M_\alpha},
\qquad k\neq\ell,\\
\Tr(R_{\alpha,k}R_{\beta,\ell})&=z_{\alpha\beta}=\frac 1d w_{\alpha}w_{\beta},\qquad \alpha\neq\beta,
\end{split}
\end{equation}
where $\sum_{\alpha=1}^NM_\alpha w_\alpha^2=d\kappa$. If we compare the above coefficients with Definition \ref{MUGETFS} of MU GETFs, we find that
\begin{equation}
a_\alpha=w_\alpha,\qquad b_\alpha=\frac{x_\alpha}{w_\alpha^2},\qquad c_\alpha=\frac{y_\alpha}{w_\alpha^2},\qquad f=\frac{z_{\alpha\beta}}{w_\alpha w_\beta},
\end{equation}
and also $\gamma_\alpha=\kappa_\alpha/w_\alpha=w_\alpha M_\alpha/d$. It is straightforward to determine the positivity of the following constants: $\gamma_\alpha>0$, $a_\alpha>0$, and $f>0$ (because $w_\alpha>0$). 
Now, for $\mathcal{R}$ to be a collection of MU GETFs, $b_\alpha$ has to belong to the range given by eq. (\ref{brange2}). This is the case as long as $\kappa_->0$ and
\begin{equation}\label{TrRa}
w_\alpha\geq
\sqrt{\frac{d\kappa_-}{M_\alpha}}\max\left\{1,\sqrt{\frac{M_\alpha-1}{d-1}}\right\}.
\end{equation}
The above condition follows from eq. (\ref{kappaa}) for $w_{\alpha,k}=w_\alpha$ and already guarantees $c_\alpha\geq 0$.
Finally, observe that $S_\alpha=a_\alpha^2(b_\alpha-c_\alpha)=\kappa_-$, which means that these MU GETFs are indeed equidistant.
\end{proof}

The class of homogeneous conical 2-designs is recovered for $R_{\alpha,k}$ that are of equal trace. Then,
\begin{equation}
w_\alpha=w=\frac{d\kappa}{d^2+N-1},\qquad x_\alpha=x=\frac{w^2}{d}+\frac{M_\alpha -1}{M_\alpha}\kappa_-,\qquad y_\alpha=y=\frac{w^2}{d}-\frac{\kappa_-}{M_\alpha}, \qquad z_{\alpha\beta}=z=\frac{w^2}{d}.
\end{equation}
One immediately notes that the constant coefficients $x$, $y$, $z$ can be organized into an increasing sequence: $y<z<x$.

At last, we combine the results of Propositions \ref{MUGETFStoconical} and \ref{conicaltoMUGETFS} to formulate our next main result.

\begin{Theorem}\label{main2}
Consider $N$ GETFs $\mathcal{R}_\alpha=\{R_{\alpha,k};\,k=1,\ldots,M_\alpha\}$ such that $\mathcal{R}=\cup_{\alpha=1}^N\mathcal{R}_\alpha$ spans the operator space $\mathcal{B}(\mathcal{H})$ and $|\mathcal{R}|=d^2+N-1$. Then, $\mathcal{R}$ is a conical 2-design if and only if it is a set of equidistant mutually unbiased generalized equiangular tight frames.
\end{Theorem}

In particular, if $\Gamma\equiv\sum_{\alpha=1}^N\gamma_\alpha=\sum_{\alpha=1}^Nw_\alpha M_\alpha/d=1$, then Theorem \ref{main2} establishes a one-to-one correspondence between a class of conical 2-designs and equidistant generalized equiangular measurements (GEAMs) \cite{GEAM}. If instead the $N$ GETFs $\mathcal{R}_\alpha$ are POVMs, then one recovers equidistant generalized symmetric measurements \cite{SIC-MUB_general}, which are simply equidistant MU GETFs with all $\gamma_\alpha=1$. 

\begin{Theorem}\label{conicalPOVM2}
Consider $N$ $(1,M_\alpha)$-POVMs $\mathcal{P}_\alpha=\{P_{\alpha,k};\,k=1,\ldots,M_\alpha\}$ such that $\mathcal{P}=\cup_{\alpha=1}^N\mathcal{P}_\alpha$ spans the operator space $\mathcal{B}(\mathcal{H})$ and $|\mathcal{P}|=d^2+N-1$. Then, $\mathcal{P}$ is a conical 2-design if and only if it is an equidistant generalized symmetric measurement.
\end{Theorem}

The known special cases include $(N,M)$-POVMs ($M_\alpha=M$) \cite{SIC-MUB} and MUMs ($M_\alpha=d$, $N=d+1$) \cite{Kalev}.

\section{Conclusions}

In this paper, we have analyzed two important classes of conical 2-designs and their relations to other known constructions: equiangular tight frames of arbitrary rank (GETFs) and their mutually unbiased collections (MU GETFs). First, we consider an informationally complete set of semi-positive operators $\mathcal{R}=\{R_k;\,k=1,\ldots,M\}$. Assuming that all $R_k$ are linearly independent, we find the most general trace relations between the elements of $\mathcal{R}$. As a special class, we consider homogeneous conical 2-designs and analyze their properties. Then, we provide the necessary and sufficient conditions for a set of linearly independent operators to be a homogeneous conical 2-design. We also give an example of $d^2$ linearly independent operators that form an inhomogeneous conical 2-design, thus proving that this set is non-empty. A one-to-one correspondence is also given for conical 2-designs constructed from POVMs. Analogical results are then obtained for an informationally overcomplete set $\mathcal{R}=\{R_{\alpha,k};\,\alpha=1,\ldots,N;\,k=1,\ldots,M_\alpha\}$, where partial sums $\sum_{k=1}^{M_\alpha}\Tr(R_{\alpha,k})R_{\alpha,k}$ are linearly proportional to the identity operator $\mathbb{I}_d$. Then, an equivalence relation is established between equidistant MU GETFs and a class of conical 2-designs that go beyond homogeneous designs. Finally, it turns out that generalized equiangular measurements are the most general conical 2-designs constructed from symmetric measurements.

In further research, it would be interesting to present a full analysis of conical 2-designs. In particular, one could provide general construction methods and more examples of inhomogeneous designs. Another open problem is to generalize the results of refs. \cite{GEAM_coherence} and \cite{GEAM_Pmaps} to wider classes of conical designs. This would help to establish whether the following properties extend to all conical 2-designs: (i) a linear relation between the index of coincidence and the purity of a mixed state, (ii) the existence of a family of measures that depend only on $\kappa_\pm$, (iii) the sum of squared elements being proportional to the identity operator. 

Projective 2-designs and their generalizations find wide applications in many subfields of quantum information. In~the foundations of quantum theory, they are used to define uncertainty relations in terms of entropies~\cite{Rastegin5,Ivanovic2,Sanchez,ChenFei,Rastegin6}. Nowadays, entropic uncertainty relations are used in most quantum cryptographic protocols, including quantum key distribution and two-party quantum cryptography~\cite{Coles,Koashi,Wehner}. In~quantum entanglement, projective and conical 2-designs play an important role in the detection and quantification of entangled states via entanglement witnesses~\cite{EW-SIC,MUBs,bound_ent,Li,MUM_purity,EW-2MUB} and criteria based on the correlation matrix~\cite{ESIC,Spengler,ChenMa,ChenLi}. Entanglement is a precious quantum resource for the purposes of quantum communication and information processing, quantum computation, and~other modern technologies~\cite{Nielsen,HHHH}. Its strength lies in ensuring better performance in many quantum tasks, like teleportation~\cite{Masanes} and channel discrimination~\cite{BaeDarek}. Another important problem is the reconstruction of a quantum state from a finite number of measurements, which in turn allows for observable estimation. This is known as state tomography and is often performed using SIC POVMs or mutually unbiased bases~\cite{ZhuEnglert,QST1,QST2,QST3,QST4,QST5,QST6,QST7}. Hence, if~special classes of conical 2-designs find relevance in the aforementioned fields, one could assume that this applicational prowess could straightforwardly extend to the entire family of conical~2-designs.

\section{Acknowledgements}

This research was funded in whole or in part by the National Science Centre, Poland, Grant number 2021/43/D/ST2/00102. For the purpose of Open Access, the author has applied a CC-BY public copyright licence to any Author Accepted Manuscript (AAM) version arising from this submission.

\appendix

\section{Properties of generalized equiangular tight frames}

To derive the values of $a$ and $c$, we use the fact that the elements of a GETF sum up to $\gamma\mathbb{I}_d$. This way, we obtain
\begin{equation}\label{tra}
d\gamma=\Tr(\gamma\mathbb{I}_d)=\sum_{k=1}^M\Tr(P_k)=Ma
\end{equation}
and
\begin{equation}\label{tra2}
a\gamma=\gamma\Tr(P_k)=\Tr(P_k\gamma\mathbb{I}_d)=\Tr(P_k^2)+\sum_{\ell\neq k} \Tr(P_kP_\ell)=a^2b+(M-1)a^2c,
\end{equation}
from which eq. (\ref{ac}) immediately follows.

The lower bound on $b$ in eq. (\ref{brange}) is calculated from the Cauchy-Schwarz inequality. It states that $[\Tr(AB)]^2\leq\Tr(A^2)\Tr(B^2)$ for any Hermitian operators $A$ and $B$, where the equality holds for $A=B$. In our case,
\begin{equation}\label{CS1}
a^2=[\Tr(P_k\mathbb{I}_d)]^2\leq\Tr(P_k^2)\Tr(\mathbb{I}_d^2)=da^2b.
\end{equation}
We exclude the equality in eq. (\ref{CS1}) because then all $P_k$ are linearly proportional to $\mathbb{I}_d$. The upper bound follows from $\Tr(P_k^2)\leq[\Tr(P_k)]^2$ and $\Tr(P_kP_\ell)\geq 0$, which are equivalent to $b\leq 1$ and $c\geq 0$, respectively.

Now, take $M$ real numbers $r_k$, $k=1,\ldots,M$. If
\begin{equation}\label{IC}
\sum_{k=1}^Mr_kP_k=0
\end{equation}
implies that all $r_k=0$, then $P_k$ are linearly independent operators. By taking the trace of eq. (\ref{IC}), we find that
\begin{equation}
\Tr\left(\sum_{k=1}^Mr_kP_k\right)=a\sum_{k=1}^Mr_k=0.
\end{equation}
If we first multiply both sides of eq. (\ref{IC}) by $P_\ell$ and then calculate the trace, we obtain
\begin{equation}
\Tr\left(\sum_{k=1}^Mr_kP_kP_\ell\right)=a^2br_\ell+a^2c\sum_{k\neq\ell}r_k
=a^2\left[(b-c)r_\ell+c\sum_{k=1}^Mr_k\right]=0.
\end{equation}
Therefore, due to $a>0$ and $b\neq c$, we finally have $r_\ell=0$.

\section{Linearly independent subset of mutually unbiased GETFs}

From definition, a set of MU GETFs cannot consist of linearly independent operators because of the constraints $\sum_{k=1}^{M_\alpha}P_{\alpha,k}=\gamma_\alpha\mathbb{I}_d$ for every GETF $\mathcal{P}_\alpha$. However, if we take $M_\alpha-1$ operators from each GETF, then they can be linearly independent from each other and the identity operator $\mathbb{I}_d$, provided that some additional conditions are met. There can be at most $d^2$ linearly independent operators on $\mathbb{C}^d$, and hence the number of all $P_{\alpha,k}$ is bounded by $|\mathcal{P}|=\sum_{\alpha=1}^NM_\alpha\leq d^2-1+N$. Moreover, it is assumed that $P_{\alpha,k}\neq\gamma_\alpha\mathbb{I}_d$, so that the number of GETF elements is $M_\alpha\geq 2$.

Now, let us define
\begin{equation}\label{PP}
\mathbb{P}=r_0\mathbb{I}_d+\sum_{\alpha=1}^N\sum_{k=1}^{M_\alpha-1}r_{\alpha,k}P_{\alpha,k}.
\end{equation}
From basic algebra, we know that if $\mathbb{P}=0$ implies $r_0=r_{\alpha,k}=0$, then the elements of $\{\mathbb{I}_d,\,P_{\alpha,k};\,k=1,\ldots,M_\alpha-1;\,\alpha=1,\ldots,N\}$ are linearly independent. Assuming that $\mathbb{P}=0$, we calculate
\begin{equation}\label{trp}
\Tr(\mathbb{P})=dr_0+\sum_{\alpha=1}^Na_\alpha\sum_{k=1}^{M_\alpha-1}r_{\alpha,k}=0,
\end{equation}
and also, for any $P_{\alpha,k}$ with $k=1,\ldots,M_\alpha-1$,
\begin{equation}\label{22}
\begin{split}
\Tr(\mathbb{P}P_{\alpha,k})&=a_\alpha r_0
+a_\alpha^2(b_\alpha-c_\alpha)r_{\alpha,k}
+a_\alpha^2(c_\alpha-f)\sum_{\ell=1}^{M_\alpha-1}r_{\alpha,\ell}
+a_\alpha f\sum_{\beta=1}^Na_\beta\sum_{\ell=1}^{M_\beta-1}r_{\beta,\ell}\\&
=a_\alpha^2(b_\alpha-c_\alpha)\left[r_{\alpha,k}
-\frac{1}{M_\alpha}\sum_{\ell=1}^{M_\alpha-1}r_{\alpha,\ell}\right]
=0,
\end{split}
\end{equation}
where we made use of eq. (\ref{trp}) and the fact that
\begin{equation}\label{caf2}
c_\alpha-f=-\frac{b_\alpha-c_\alpha}{M_\alpha}.
\end{equation}
Next, if one takes the sum over all $k=1,\ldots,M_\alpha-1$ in eq. (\ref{22}), it results in
\begin{equation}\label{SS}
\sum_{k=1}^{M_\alpha-1}\Tr(\mathbb{P}P_{\alpha,k})
=\frac{a_\alpha^2(b_\alpha-c_\alpha)}{M_\alpha}\sum_{k=1}^{M_\alpha-1}r_{\alpha,k}
=0,
\end{equation}
and therefore $\sum_{k=1}^{M_\alpha-1}r_{\alpha,k}=0$. Substituting this into eqs. (\ref{trp}) and (\ref{22}) gives $r_0=0$ and $r_{\alpha,k}=0$, respectively.

\section{Admissible range of $S$}

Using the formulas for $a_\alpha$ and $c_\alpha$ from eq. (\ref{ac2}), we calculate
\begin{equation}
S_\alpha=a_\alpha^2(b_\alpha-c_\alpha)=\frac{d\gamma_\alpha^2}{M_\alpha}\frac{db_\alpha-1}{M_\alpha-1}.
\end{equation}
Now, through performing linear operations on eq. (\ref{brange2}) that do not change the inequality sign, it is straightforward to show that
\begin{equation}
0<S_\alpha\leq\gamma_\alpha^2\min\left\{\frac{d(d-1)}{M_\alpha(M_\alpha-1)},\frac{d}{M_\alpha}\right\}.
\end{equation}
The range for $S_\alpha\equiv S$ is recovered when one takes the minimization of the right hand-side of this inequality over all $\alpha$.

\section{Trace parameters for conical 2-designs in Proposition 6}

From eq. (\ref{trrr}), we know that, for all $k$ and $\alpha$,
\begin{equation}\label{gen}
F(\alpha,k)\equiv(x_{\alpha,k}-\kappa_-)R_{\alpha,k}
+\sum_{\ell\neq k}y_{\alpha,k,\ell}R_{\alpha,\ell}
+\sum_{\beta\neq\alpha}\sum_{\ell=1}^{M_\beta}z_{\alpha,k;\beta,\ell}
R_{\beta,\ell}-w_{\alpha,k}\kappa_+\mathbb{I}_d=0.
\end{equation}
To solve this equation, we need to rewrite it using only $d^2$ linearly independent operators. 
Recall that the operators $R_{\alpha,M_\alpha}$ are linearly dependent on $\mathbb{I}_d$ and $R_{\alpha,k}$, $k=1,\ldots,M_\alpha-1$, through eq. (\ref{kappas}). Hence,
\begin{equation}\label{RaMa}
R_{\alpha,M_\alpha}=\frac{1}{w_{\alpha,M_\alpha}}\left(\kappa_\alpha\mathbb{I}_d
-\sum_{k=1}^{M_\alpha-1}w_{\alpha,k}R_{\alpha,k}\right).
\end{equation}
Applying eq. (\ref{RaMa}) into eq. (\ref{gen}) results in two formulas:
for all $k=1,\ldots,M_\alpha-1$,
\begin{equation}\label{gen1}
\begin{split}
F(\alpha,k)=&\mathbb{I}_d\left(-w_{\alpha,k}\kappa_++y_{\alpha,k,M_\alpha}\frac{\kappa_\alpha}
{w_{\alpha,M_\alpha}}+\sum_{\beta\neq\alpha}z_{\alpha,k;\beta,M_\beta}\frac{\kappa_\beta}
{w_{\beta,M_\beta}}\right)+R_{\alpha,k}\left(x_{\alpha,k}-\kappa_--y_{\alpha,k,M_\alpha}
\frac{w_{\alpha,k}}{w_{\alpha,M_\alpha}}\right)\\&+\sum_{\ell\neq k,M_\alpha}
R_{\alpha,\ell}\left(y_{\alpha,k,\ell}-y_{\alpha,k,M_\alpha}\frac{w_{\alpha,\ell}}
{w_{\alpha,M_\alpha}}\right)+\sum_{\beta\neq\alpha}\sum_{\ell=1}^{M_\beta-1}
\left(z_{\alpha,k;\beta,\ell}-z_{\alpha,k;\beta,M_\beta}\frac{w_{\beta,\ell}}
{w_{\beta,M_\beta}}\right),
\end{split}
\end{equation}
and, for $k=M_\alpha$,
\begin{equation}\label{gen2}
\begin{split}
F(\alpha,M_\alpha)=&\mathbb{I}_d\left(-w_{\alpha,M_\alpha}\kappa_++(x_{\alpha,M_\alpha}-\kappa_-)
\frac{\kappa_\alpha}
{w_{\alpha,M_\alpha}}+\sum_{\beta\neq\alpha}z_{\alpha,M_\alpha;\beta,M_\beta}\frac{\kappa_\beta}
{w_{\beta,M_\beta}}\right)\\&+\sum_{\ell=1}^{M_\alpha-1}
R_{\alpha,\ell}\left(y_{\alpha,M_\alpha,\ell}-(x_{\alpha,M_\alpha}-\kappa_-)\frac{w_{\alpha,\ell}}
{w_{\alpha,M_\alpha}}\right)+\sum_{\beta\neq\alpha}\sum_{\ell=1}^{M_\beta-1}
\left(z_{\alpha,M_\alpha;\beta,\ell}-z_{\alpha,M_\alpha;\beta,M_\beta}\frac{w_{\beta,\ell}}
{w_{\beta,M_\beta}}\right).
\end{split}
\end{equation}
Now, the solution of $F(\alpha,k)=0$ for all $k=1,\ldots,M_\alpha$ and $\alpha=1,\ldots,N$ is equivalent to all the real coefficients that multiply each operators vanishing. In other words,
\begin{equation*}
\left\{\begin{alignedat}{3}
z_{\alpha,k;\beta,\ell}&=z_{\alpha,k;\beta,M_\beta}\frac{w_{\beta,\ell}}{w_{\beta,M_\beta}},\qquad &&{\rm(i.a)}\\
w_{\alpha,k}\kappa_+&=y_{\alpha,k,M_\alpha}\frac{\kappa_\alpha}{w_{\alpha,M_\alpha}}
+\sum_{\beta\neq\alpha}z_{\alpha,k;\beta,M_\beta}\frac{\kappa_\beta}{w_{\beta,M_\beta}},\qquad &&{\rm(ii.a)}\\
y_{\alpha,k,M_\alpha}&=(x_{\alpha,k}-\kappa_-)\frac{w_{\alpha,M_\alpha}}{w_{\alpha,k}},\qquad &&{\rm(iii.a)}\\
y_{\alpha,k,j}&=y_{\alpha,k,M_\alpha}\frac{w_{\alpha,j}}{w_{\alpha,M_\alpha}},\qquad &&{\rm(iv.a)}
\end{alignedat}\right.
\end{equation*}
and
\begin{equation*}
\left\{\begin{alignedat}{3}
z_{\alpha,M_\alpha;\beta,\ell}&=z_{\alpha,M_\alpha;\beta,M_\beta}\frac{w_{\beta,\ell}}{w_{\beta,M_\beta}},\qquad &&{\rm(i.b)}\\
w_{\alpha,M_\alpha}\kappa_+&=(x_{\alpha,M_\alpha}-\kappa_-)\frac{\kappa_\alpha}{w_{\alpha,M_\alpha}}
+\sum_{\beta\neq\alpha}z_{\alpha,M_\alpha;\beta,M_\beta}\frac{\kappa_\beta}{w_{\beta,M_\beta}},\qquad &&{\rm(ii.b)}\\
y_{\alpha,M_\alpha,j}&=(x_{\alpha,M_\alpha}-\kappa_-)\frac{w_{\alpha,j}}{w_{\alpha,M_\alpha}},\qquad &&{\rm(iii.b)}
\end{alignedat}\right.
\end{equation*}
where $k,j=1,\ldots,M_\alpha-1$, $\ell=1,\ldots,M_\beta-1$, and $\beta\neq\alpha$. Now, eqs. (i.a) and (i.b), together with eq. (\ref{cc3}), give
\begin{equation}
\bigforall_{\alpha,k,\beta\neq\alpha,\ell}\quad z_{\alpha,k;\beta,\ell}=\frac 1d w_{\alpha,k}w_{\beta,\ell}.
\end{equation}
Next, if we input this result and eq. (iii.b) into eqs. (ii.b) and (ii.a), we recover
\begin{equation}
\bigforall_{\alpha,j<M_\alpha}\quad y_{\alpha,j,M_\alpha}
=\frac{\kappa_\alpha-\kappa_-}{d\kappa_\alpha}w_{\alpha,j}w_{\alpha,M_\alpha}.
\end{equation}
Combining the above formula with eq. (iv.a) returns the general relation
\begin{equation}
\bigforall_{\alpha,k,j\neq k}\quad y_{\alpha,k,j}
=\frac{\kappa_\alpha-\kappa_-}{d\kappa_\alpha}w_{\alpha,k}w_{\alpha,j}.
\end{equation}
Finally, substituting the results for $y_{\alpha,k,j}$ into eqs. (iii.a) and (iii.b), one arrives at
\begin{equation}
\bigforall_{\alpha,k}\quad x_{\alpha,k}
=\kappa_-+\frac{\kappa_\alpha-\kappa_-}{d\kappa_\alpha}w_{\alpha,k}^2.
\end{equation}

\section{Consequences of semi-positivity of $R_{\alpha,k}$}

Conical 2-designs are constructed from positive operators $R_{\alpha,k}$, and hence $\Tr(R_{\alpha,k})=w_{\alpha,k}>0$ as well as
\begin{align}
\Tr(R_{\alpha,k}R_{\alpha,\ell})&=y_{\alpha,k,\ell}
=\frac{\kappa_\alpha-\kappa_-}{d\kappa_\alpha}w_{\alpha,k}w_{\alpha,j}\geq 0,
\qquad \ell\neq k,\label{con1}\\
\Tr(R_{\alpha,k}R_{\beta,\ell})&=z_{\alpha,k;\beta,\ell}=\frac 1d w_{\alpha,k}w_{\beta,\ell}\geq 0,\qquad \beta\neq\alpha.\label{con2}
\end{align}
Note that eq. (\ref{con2}) immediately holds for $w_{\alpha,k}>0$, whereas eq. (\ref{con1}) is satisfied under the additional condition that $\kappa_\alpha\geq\kappa_-$. The admissible range for
\begin{equation}
\Tr(R_{\alpha,k}^2)=x_{\alpha,k}=\kappa_-+\frac{\kappa_\alpha-\kappa_-}{d\kappa_\alpha}w_{\alpha,k}^2
\end{equation}
is determined by
\begin{equation}\label{kwadrat}
\Tr(R_{\alpha,k}^2)\leq[\Tr(R_{\alpha,k})]^2 \qquad\Longleftrightarrow\qquad x_{\alpha,k}\leq w_{\alpha,k}^2
\end{equation}
and the Cauchy-Schwarz inequality
\begin{equation}\label{CS3}
w_{\alpha,k}^2=[\Tr(R_{\alpha,k}\mathbb{I}_d)]^2\leq\Tr(R_{\alpha,k}^2)\Tr(\mathbb{I}_d^2)
=dx_{\alpha,k}.
\end{equation}
Eq. (\ref{kwadrat}) simplifies to
\begin{equation}
w_{\alpha,k}^2\geq \frac{d\kappa_-\kappa_\alpha}{(d-1)\kappa_\alpha+\kappa_-},
\end{equation}
and from eq. (\ref{CS3}) we get
\begin{equation}
w_{\alpha,k}^2<d\kappa_\alpha,
\end{equation}
where the equality is excluded because it holds only for $R_{\alpha,k}=w_{\alpha,k}\mathbb{I}_d/d$.

\bibliography{C:/Users/cyndaquilka/OneDrive/Fizyka/bibliography}

\begin{thebibliography}{10}
\providecommand{\url}[1]{\texttt{#1}}
\providecommand{\urlprefix}{URL }
\providecommand{\eprint}[2][]{\url{#2}}

\bibitem{Zhou}
P.~Zhou, J. Phys. A: Math. Theor. \textbf{45}, 215305 (2012).

\bibitem{Song}
J.-F. Song and Z.-Y. Wang, Int. J. Theor. Phys. \textbf{50}, 2410 (2011).

\bibitem{Bouchard}
F.~Bouchard, K.~Heshami, D.~England, R.~Fickler, R.~W. Boyd, B.-G. Englert,
  L.~L. S{\'a}nchez-Soto, and E.~Karimi, Quantum \textbf{2}, 111 (2018).

\bibitem{IOC1}
H.~Zhu, Phys. Rev. A \textbf{90}, 012115 (2014).

\bibitem{Adamson}
R.~B.~A. Adamson and A.~M. Steinberg, Phys. Rev. Lett. \textbf{105}, 030406
  (2010).

\bibitem{Zhu_QSE}
H.~Zhu, PRX Quantum \textbf{3}, 030306 (2022).

\bibitem{ESIC}
J.~Shang, A.~Asadian, H.~Zhu, and O.~G{\"{u}}hne, Phys. Rev. A \textbf{98},
  022309 (2018).

\bibitem{EW-SIC}
T.~Li, L.-M. Lai, D.-F. Liang, S.-M. Fei, and Z.-X. Wang, Int. J. Theor. Phys.
  \textbf{59}, 3549--3557 (2020).

\bibitem{KalevBae}
A.~Kalev and J.~Bae, Phys. Rev. A \textbf{87}, 062314 (2013).

\bibitem{Blume}
R.~Blume-Kohout, J.~O.~S. Yin, and S.~J. van Enk, Phys. Rev. Lett.
  \textbf{105}, 170501 (2010).

\bibitem{SM_Pmaps}
J.~Li, H.~Yao, S.-M. Fei, Z.~Fan, and H.~Ma, Phys. Rev. A \textbf{109}, 052426
  (2024).

\bibitem{Lai2}
L.-M. Lai, T.~Li, S.-M. Fei, and Z.-X. Wang, Quant. Inf. Proc. \textbf{19}, 93
  (2020).

\bibitem{OperationalSICs}
A.~Tavakoli, M.~Farkas, D.~Rosset, J.-D. Bancal, and J.~Kaniewski, Sci. Adv.
  \textbf{7}, eabc3847 (2021).

\bibitem{Bene}
T.~V{\'e}rtesi and E.~Bene, Phys. Rev. A \textbf{82}, 062115 (2010).

\bibitem{Prugovecki}
E.~Prugove{\v{c}}ki, Int. J. Theor. Phys. \textbf{16}, 321--331 (1977).

\bibitem{PetzRuppert}
D.~Petz and L.~Ruppert, Rep. Math. Phys. \textbf{69}, 161 (2012).

\bibitem{Pimenta}
W.~M. Pimenta, B.~Marques, T.~O. Maciel, R.~O. Vianna, A.~Delgado, C.~Saavedra,
  and S.~P{\'a}dua, Phys. Rev. A \textbf{88}, 012112 (2013).

\bibitem{Bent}
N.~Bent, H.~Qassim, A.~A. Tahir, D.~Sych, G.~Leuchs, L.~L. S{\'a}nchez-Soto,
  E.~Karimi, and R.~W. Boyd, Phys. Rev. X \textbf{5}, 041006(R) (2015).

\bibitem{Renes}
J.~M. Renes, R.~Blume-Kohout, A.~J. Scott, and C.~M. Caves, J. Math. Phys.
  \textbf{45}, 2171 (2004).

\bibitem{Schwinger}
J.~Schwinger, Proc. Nat. Acad. Sci. U.S.A. \textbf{46}, 570 (1960).

\bibitem{Szarek}
M.~B. Ruskai, S.~Szarek, and E.~Werner, Linear Algebra Appl. \textbf{347(1-3)},
  159--187 (2002).

\bibitem{Neumaier}
A.~Neumaier, Eindhoven University of Technology: Dept of Mathematics:
  memorandum \textbf{8109}, 98 (1981).

\bibitem{Hoggar}
S.~Hoggar, Eur. J. Comb. \textbf{3}, 233--254 (1982).

\bibitem{Zauner}
G.~Zauner, \textit{Quantum Designs Foundations of a non-commutative Design
  Theory}, Ph.D. thesis, University of Vienna (1999).

\bibitem{Scott}
A.~J. Scott, J. Phys. A: Math. Gen. \textbf{39}, 13507 (2006).

\bibitem{semi-SIC}
I.~J. Geng, K.~Golubeva, and G.~Gour, Phys. Rev. Lett. \textbf{126}, 100401
  (2021).

\bibitem{EOM22}
L.~Feng and S.~Luo, Phys. Lett. A \textbf{445}, 128243 (2022).

\bibitem{EOM24}
L.~Feng, S.~Luo, Y.~Zhao, and Z.~Guo, Phys. Rev. A \textbf{109}, 012218 (2024).

\bibitem{EOMq3}
Y.~Zhao, Z.~Guo, L.~Feng, S.~Luo, and T.-L. Lee, Phys. Lett. A \textbf{495},
  129314 (2024).

\bibitem{EOM25}
Z.~Guo, Y.~Liu, T.-L. Lee, and S.~Luo, Phys. Rev. A \textbf{111}, 012430
  (2025).

\bibitem{SIC-MUB}
K.~Siudzi{\'n}ska, Phys. Rev. A \textbf{105}, 042209 (2022).

\bibitem{SIC-MUB_general}
K.~Siudzi{\'n}ska, J. Phys. A: Math. Theor. \textbf{57}, 355301 (2024).

\bibitem{GEAM}
K.~Siudzi{\'n}ska, J. Phys. A: Math. Theor. \textbf{57}, 335302 (2024).

\bibitem{SICMUB_design}
F.~Huang, F.~Wu, L.~Tang, Z.-W. Mo, and M.-Q. Bai, Phys. Scr. \textbf{98},
  105103 (2023).

\bibitem{Graydon}
M.~A. Graydon and D.~M. Appleby, J. Phys. A: Math. Theor. \textbf{49}, 085301
  (2016).

\bibitem{Graydon2}
M.~A. Graydon and D.~M. Appleby, J. Phys. A: Math. Theor. \textbf{49}, 33LT02
  (2016).

\bibitem{GraydonPhD}
M.~A. Graydon, \textit{Conical Designs and Categorical Jordan Algebraic
  Post-Quantum Theories}, Ph.D. thesis, University of Waterloo (2017).

\bibitem{GEAM_Pmaps}
K.~Siudzi\'{n}ska, Sci. Rep. \textbf{15}, 29890 (2025).

\bibitem{GEAM_coherence}
K.~Siudzi\'{n}ska, J. Phys. A: Math. Theor. \textbf{58}, 375302 (2025).

\bibitem{Choi}
M.-D. Choi, Linear Algebra Appl. \textbf{10}, 285--290 (1975).

\bibitem{Jamiolkowski}
A.~Jamio{\l}kowski, Rep. Math. Phys. \textbf{3}, 275--278 (1972).

\bibitem{Gour}
A.~Kalev and G.~Gour, J. Phys. A: Math. Theor. \textbf{47}, 335302 (2014).

\bibitem{Yoshida}
M.~Yoshida and G.~Kimura, Phys. Rev. A \textbf{106}, 022408 (2022).

\bibitem{Kalev}
A.~Kalev and G.~Gour, New J. Phys. \textbf{16}, 053038 (2014).

\bibitem{Wang}
K.~Wang, N.~Wu, and F.~Song, Phys. Rev. A \textbf{98}, 032329 (2018).

\bibitem{Strohmer2}
T.~Strohmer and R.~Heath, Appl. Comput. Harmon. Anal. \textbf{14}, 257--275
  (2003).

\bibitem{Strohmer}
T.~Strohmer, Linear Algebra Appl. \textbf{429}, 326 (2008).

\bibitem{Lemmens}
P.~W.~H. Lemmens and J.~J. Seidel, J. Algebra \textbf{24}, 494 (1973).

\bibitem{Fickus}
M.~Fickus and B.~R. Mayo, \textit{Mutually Unbiased Equiangular Tight Frames}
  (2020), arXiv:2001.02055 [math.FA].

\bibitem{Goyeneche2}
F.~C. Perez, V.~G. Avella, and D.~Goyeneche, Quantum \textbf{6}, 851 (2022).

\bibitem{Rastegin5}
A.~E. Rastegin, Eur. Phys. J. D \textbf{67}, 269 (2013).

\bibitem{Orlowski}
L.~Kn{\"o}ll and A.~Or{\l}owski, Phys. Rev. A \textbf{51}, 1622 (1995).

\bibitem{Ivanovic2}
I.~D. Ivanovi{\'{c}}, J. Phys. A: Math. Gen. \textbf{25}, L363 (1992).

\bibitem{Sanchez}
J.~S{\'{a}}nchez, Phys. Lett. A \textbf{173}, 233--239 (1993).

\bibitem{ChenFei}
B.~Chen and S.-M. Fei, Quant. Inf. Proc. \textbf{14}, 2227--2238 (2015).

\bibitem{Rastegin6}
A.~E. Rastegin, Phys. Scr. \textbf{89}, 085101 (2014).

\bibitem{Coles}
P.~J. Coles, M.~Berta, M.~Tomamichel, and S.~Wehner, Rev. Mod. Phys.
  \textbf{89}, 015002 (2017).

\bibitem{Koashi}
M.~Koashi, New J. Phys. \textbf{11}, 045018 (2009).

\bibitem{Wehner}
S.~Wehner and A.~Winter, New J. Phys. \textbf{12}, 025009 (2010).

\bibitem{MUBs}
D.~Chru{\'s}ci{\'n}ski, G.~Sarbicki, and F.~A. Wudarski, Phys. Rev. A
  \textbf{97(12)}, 032318 (2018).

\bibitem{bound_ent}
J.~Bae, A.~Bera, D.~Chru{\'s}ci{\'n}ski, B.~C. Hiesmayr, and D.~McNulty, J.
  Phys. A: Math. Theor. \textbf{55}, 505303 (2022).

\bibitem{Li}
T.~Li, L.-M. Lai, S.-M. Fei, and Z.-X. Wang, Int. J. Theor. Phys. \textbf{58},
  3973--3985 (2019).

\bibitem{MUM_purity}
M.~Salehi, S.~J. Akhtarshenas, M.~Sarbishaei, and H.~Jaghouri, Quantum Inf.
  Process. \textbf{20}, 401 (2021).

\bibitem{EW-2MUB}
K.~Wang and Z.-J. Zheng, Int. J. Theor. Phys. \textbf{60}, 274--283 (2021).

\bibitem{Spengler}
C.~Spengler, M.~Huber, S.~Brierley, T.~Adaktylos, and B.~C. Hiesmayr, Phys.
  Rev. A \textbf{86}, 022311 (2012).

\bibitem{ChenMa}
B.~Chen, T.~Ma, and S.-M. Fei, Phys. Rev. A \textbf{89}, 064302 (2014).

\bibitem{ChenLi}
B.~Chen, T.~Li, and S.-M. Fei, Quant. Inf. Proc. \textbf{14}, 2281--2290
  (2015).

\bibitem{Nielsen}
M.~A. Nielsen and I.~L. Chuang, \textit{Quantum Computation and Quantum
  Information},  Cambridge University Press, Cambridge 2010.

\bibitem{HHHH}
R.~Horodecki, P.~Horodecki, M.~Horodecki, and K.~Horodecki, Rev. Mod. Phys.
  \textbf{81}, 865 (2009).

\bibitem{Masanes}
L.~Masanes, Phys. Rev. Lett. \textbf{96}, 150501 (2006).

\bibitem{BaeDarek}
J.~Bae, D.~Chru{\'s}ci{\'n}ski, and M.~Piani, Phys. Rev. Lett. \textbf{122},
  140404 (2019).

\bibitem{ZhuEnglert}
H.~Zhu and B.-G. Englert, Phys. Rev. A \textbf{84}, 022327 (2011).

\bibitem{QST1}
H.~Yuan, X.~Shi, and L.-F. Han, Res. Phys. \textbf{34}, 105228 (2022).

\bibitem{QST2}
J.~Diaz-Guevara, I.~Sainz, and A.~B. Klimov, J. Phys. A: Math. Theor.
  \textbf{54}, 295305 (2021).

\bibitem{QST3}
I.~Sainz, A.~Garcia, and A.~B. Klimov, Phys. Lett. A \textbf{382}, 66--71
  (2018).

\bibitem{QST4}
J.~Bae, B.~C. Hiesmayr, and D.~McNulty, New J. Phys. \textbf{21}, 013012
  (2019).

\bibitem{QST5}
A.~Fern{\'a}ndez-P{\'e}rez, A.~B. Klimov, and C.~Saavedra, Phys. Rev. A
  \textbf{83}, 052332 (2011).

\bibitem{QST6}
S.~N. Filippov and V.~I. Man'ko, Phys. Scr. \textbf{2011}, 014010 (2011).

\bibitem{QST7}
G.~Lima, L.~Neves, R.~Guzm{\'a}n, E.~S. G{\'o}mez, W.~A.~T. Nogueira,
  A.~Delgado, A.~Vargas, and C.~Saavedra, Opt. Express \textbf{19}, 3542--3552
  (2011).

\end{thebibliography}
\bibliographystyle{C:/Users/cyndaquilka/OneDrive/Fizyka/beztytulow2}

\end{document}